\documentclass[3p,twocolumn,preprint]{elsarticle}
\pdfoutput=1
\usepackage[utf8]{inputenc}
\usepackage[english]{babel}
\usepackage[T1]{fontenc}
\usepackage{amsmath}
\usepackage{hyperref}
\usepackage{tikz}
\usepackage{cuted,lipsum}
\usepackage{times,amsmath,amsfonts,amssymb,latexsym,amsthm}
\usepackage{graphicx,epsf}
\usepackage{subfigure}

\biboptions{sort&compress}
\usetikzlibrary{quantikz}
\newcommand{\be}{\begin{equation}}
\newcommand{\ee}{\end{equation}}
\newcommand{\ba}{\begin{eqnarray}}
\newcommand{\ea}{\end{eqnarray}}
\newcommand\tr{{\operatorname{tr}}}
\newcommand{\ignore}[1]{}

\newcommand{\haar}[0]{\mathcal U(d)}

\newcommand{\de}[0]{{\operatorname{d}}}

\newcommand{\parent}[1]{\left( {#1} \right)}
\newcommand{\aver}[1]{ \left\langle  {#1}  \right\rangle }

\newcommand{\spann}[1]{\mathrm{span}\left\{{#1}\right\}}

\newcommand{\pur}{\operatorname{Pur}}

\def\CC{{\rm\kern.24em \vrule width.04em height1.46ex depth-.07ex
    \kern-.29em C}}
\def\P{{\rm I\kern-.25em P}}
\def\RR{{\rm
         \vrule width.04em height1.58ex depth-.0ex
         \kern-.04em R}}
\def\bbbone{{\mathchoice {\rm 1\mskip-4mu l} {\rm 1\mskip-4mu l}
{\rm 1\mskip-4.5mu l} {\rm 1\mskip-5mu l}}}
\def\bbbc{{\mathchoice {\setbox0=\hbox{$\displaystyle\rm C$}\hbox{\hbox
to0pt{\kern0.4\wd0\vrule height0.9\ht0\hss}\box0}}
{\setbox0=\hbox{$\textstyle\rm C$}\hbox{\hbox
to0pt{\kern0.4\wd0\vrule height0.9\ht0\hss}\box0}}
{\setbox0=\hbox{$\scriptstyle\rm C$}\hbox{\hbox
to0pt{\kern0.4\wd0\vrule height0.9\ht0\hss}\box0}}
{\setbox0=\hbox{$\scriptscriptstyle\rm C$}\hbox{\hbox
to0pt{\kern0.4\wd0\vrule height0.9\ht0\hss}\box0}}}}
\def\bbbz{{\mathchoice {\hbox{$\sf\textstyle Z\kern-0.4em Z$}}
{\hbox{$\sf\textstyle Z\kern-0.4em Z$}}
{\hbox{$\sf\scriptstyle Z\kern-0.3em Z$}}
{\hbox{$\sf\scriptscriptstyle Z\kern-0.2em Z$}}}}

\usepackage{hyperref}
\newtheorem{theorem}{Theorem}
\newtheorem{corollary}{Corollary}
\newtheorem{lemma}{Lemma}
\newtheorem{remark}{Remark}

\newtheorem{definition}{Definition}
\newtheorem{prop}{Proposition}
\newcommand{\circuitC}[1]{\gate[style={fill=teal},label style=white,5]{C_{#1}}}

\begin{document}
\begin{frontmatter}
\title{Transitions in entanglement complexity in random quantum circuits by measurements}
\author[1]{Salvatore F.E. Oliviero\corref{cor1}%
}
\ead{s.oliviero001@umb.edu}
\cortext[cor1]{Corresponding author}
\author[1]{Lorenzo Leone}
\author[1]{Alioscia Hamma}
\address[1]{Physics Department,  University of Massachusetts Boston,  02125, USA}

\begin{abstract}
Random Clifford circuits doped with non Clifford gates exhibit transitions to universal entanglement spectrum statistics\cite{zhou2020single} and quantum chaotic behavior. In \cite{leone2021quantum} we proved that the injection of $\Omega(n)$ non Clifford gates into a $n$-qubit Clifford circuit drives the transition towards the universal value of the purity fluctuations. In this paper, we show that doping a Clifford circuit with $\Omega(n)$ single qubit non Clifford measurements is both necessary and sufficient to drive the transition to universal fluctuations of the purity.  
\end{abstract}
\end{frontmatter}

\section{Introduction}
Random unitary operators\cite{Emerson2003random} are frequently used,
in many-body physics, to model quantum chaotic behavior of highly complex Hamiltonians. For example, in the context of black holes\cite{shenker2014multi,hayden2007black}, they have been used to model the fast scrambling behavior\cite{lashkari2013towards} through the fast decay of the out-of-time order correlators\cite{shenker2014black,maldacena2016bound,roberts2017chaos}. Random quantum circuits  are also employed in many quantum information protocols, a famous example is provided by the randomized benchmarking protocol\cite{Emerson2005noise,Magesan2012bench,Knill2008randomized}, which attempts to estimate the error rate of quantum unitary operations. The simulation of random unitary operators on a classical computer requires exponential resources, while, thanks to the Gottesman-Knill theorem\cite{Gottesman:1998hu}, unitary operators belongings to a subgroup of the unitary group, the Clifford group, can be efficiently implemented on a classical computer. This result opens the question of whether the Clifford group is enough to simulate the average behavior of the unitary group and naturally leads to the concept of $t$-designs\cite{divincenzodatahiding,dankert2005random,Gross2007unitary}, i.e ensembles of unitary operators able to reproduce up to the $t$ moment of the Haar distribution over the full Unitary group $\mathcal{U}(d)$. For a formal definition of $t-$design( see \ref{app:mathpre}). It has been proven that the multiqubit Clifford group forms a 3-design\cite{webb2016clifford,zhu2017multiqubit} and thus, while it can reproduce the universal average value of the purity in a subsystem, the Clifford group does not reproduce the universal purity fluctuations. This reflects the presence of a complexity gap\cite{yang2017entanglement}: while the output state of a Clifford circuit, exhibiting a nearly maximal entanglement entropy, can be efficiently disentangled\cite{chamon2014emergent}, the output state of a universal circuit cannot. Classes of entanglement complexity can be defined as the adherence to different universal features of entanglement. A first class corresponds to that of the average entanglement in the Hilbert space, which is well known to be close to that of the maximally entangled state\cite{page1993average, lloyd1988black}. A second class corresponds to the universal behavior of the entanglement spectrum statistics (ESS), which contains richer information than just the entanglement entropy. Universal purity fluctuations, reproduced at least by a $4$-design, are probes to a third class of entanglement complexity, which is also a mark of quantum chaos: the simulation of quantum chaos requires at least a unitary $4$-design\cite{leone2020isospectral,Oliviero2020random}. The hierarchy of these different classes is an open problem: for instance, it is not known if possessing universal ESS implies universal fluctuations of the entanglement entropy.
 Starting from recent works\cite{zhou2020single,gross2020quantum}, in \cite{leone2021quantum} we proved that $\Omega(n)$ non Clifford gates randomly inserted in a Clifford circuit are both necessary and sufficient to obtain a $4$-design with an exponentially small error and thus to drive the transition towards the third class of entanglement complexity.

In this paper, we ask the question of whether one can obtain a similar transition by measurements. This paper shows that doping a Clifford circuit with one shot projective measurements drives the transition in entanglement complexity: while measurements in the Clifford basis, i.e. any basis obtained from the computational basis by Clifford rotations, are not able to drive any transition, measurements on a non-Clifford basis are. To detect the transition in entanglement complexity we look at the scaling of the purity fluctuations in a given bipartition of the Hilbert space, because, as explained above, the value of purity fluctuations discriminates between a $3-$design and a $4-$design.
We compute the average purity and its fluctuations in a Random Measurements Doped Clifford (RMDC) circuit and show that the ensemble purity fluctuations, being $\Omega(d^{-1})$ for Clifford, attain the full unitary group value $\Omega(d^{-2})$ after $\Omega(n)$ single qubit measurements. These results, together with the Gottesman-Knill theorem, are telling us, once again, that quantum chaos cannot be efficiently simulated classically, even when we have access to measurements. We also remark that collecting the outcomes from repeated measurements cannot drive any transition in the resulting mixed state: we analyze the protocol with repeated measurements, showing that after $\Omega(n)$ measurements, the output state behaves like the completely mixed state up to an arbitrarily small error scaling like $\sim d^{-\alpha} $.
\begin{figure}[h!]
\centering
\resizebox{6cm}{!}{
\begin{quantikz}
&\circuitC{1}&\qw  & \circuitC{2} &\qw  &\circuitC{3} & \meter{}& \circuitC{4}&\qw\\
& &\meter{} & &\qw & &\qw &&\qw\\
& &\qw & &\qw & &\qw &&\qw\\
& &\qw & &\meter{} & &\qw &&\qw\\
& &\qw & &\qw & &\qw &&\qw
\end{quantikz}}
\caption{Scheme of a Random Measurements Doped Clifford circuit.}
\label{scheme1}
\end{figure}
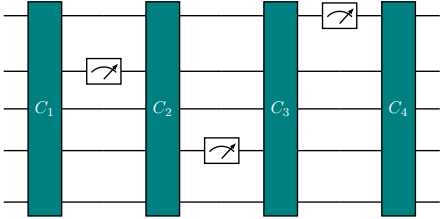
\section{Random measurements doped Clifford circuits}\label{sec:setup}
Let us consider system of $n$ qubits with Hilbert space $\mathcal{H}=\bigotimes_{j=1}^{n}\mathbb{C}_{j}^{2}$, of dimension $d\equiv\dim\mathcal{H}=2^{n}$ and $\psi\in\mathcal{B}(\mathcal{H})$ an initial reference pure state; throughout the paper $\mathcal{B}(\mathcal{H})$ denotes the linear space of linear bounded operator on $\mathcal{H}$. Let $\mbox{P}_{n}$ be the group of all $n-$qubit Pauli strings with phases $\pm 1$ and $\pm i$. Then, the Clifford group $\mathcal{C}(d)$ is the group of unitary operators which transforms Pauli strings in Pauli strings, i.e. for any $C\in\mathcal{C}(d)$ one has $C P C^{\dag}\in\mbox{P}_{n}$ for all $P\in\mbox{P}_{n}$. Let us describe the architecture of a Random Measurements Doped Clifford (RMDC) circuit sketched in Fig. \ref{scheme1}: layers of random Clifford circuits are interleaved with single qubit one shot measurements described by local projector $P_{i}$ applied on the $i$-th qubit. We denote by $k$ the number of layers in the circuit. As we will see, the positioning $i$ of the measurements will not play any role. More precisely, let us first evolve $\psi$ with a Clifford circuit $C_{0}\in\mathcal{C}(d)$:
\be
\psi\mapsto \psi_{0}\equiv C_{0}\psi C^{\dag}_{0}
\ee
then apply a single qubit one measurement on the $i_1$-th qubit: let $B_{i_1}=\spann{\ket{\phi_1}_{i_1},\ket{\phi_2}_{i_1}}$ be a single qubit basis on the $i_1$-th qubit, where $\ket{\phi_{1}}_{i_1}$ and $\ket{\phi_{2}}_{i_1}$ are kets in $\mathbb{C}_{i_1}^{2}$.  Let $P_{i}^{(\gamma)}\equiv \ket{\phi_\gamma}\bra{\phi_\gamma}_{i}\otimes \bbbone^{(\bar{i})}$, with $\gamma=1,2$, where $\bbbone^{(\bar{i})}\in\mathcal{B}(\mathcal{H}^{(\bar{i})})$ and $\mathcal{H}^{(\bar{i})}\equiv \bigotimes_{j\neq i}\mathbb{C}_{j}^{2}$. The effect of a measurement on $\psi_0$ followed by another Clifford evolution $C_1\in\mathcal{C}(d)$ is given by the following map:
\be
\psi_{0}\mapsto \psi_{1}\equiv \frac{C_{1}P_{i_1}^{(\gamma_1)}\psi_{C_0}P_{i_1}^{(\gamma_1)}C_{1}^{\dag}}{\tr(P_{i_1}^{(\gamma_1)}\psi_{C_0})}
\ee
Being one shot measurements, the outcome $\gamma_{1}$ comes random with probability $\tr(P_{i_{1}}^{(\gamma_1)}\psi_{C_0})$. Then applying $k$ Clifford circuits $C_{\alpha}$ interleaved by $k$ projective measurements in the basis $B_{i_{\alpha}}$ on the qubit $i_\alpha$, for $\alpha=1,\dots,k$:
\be
\psi_{k}\equiv\frac{C_{k}P_{i_k}^{(\gamma_k)}\cdots C_{1}P_{i_1}^{(\gamma_1)}C_{0}\psi C_{0}^{\dag}P_{i_1}^{(\gamma_1)}C_{1}^{\dag}\cdots P_{i_k}^{(\gamma_k)}C_{k}^{\dag}}{\tr(P_{i_k}^{(\gamma_k)}\cdots C_{1}P_{i_1}^{(\gamma_1)}C_{0}\psi C_{0}^{\dag}P_{i_1}^{(\gamma_1)}C_{1}^{\dag}\cdots P_{i_k}^{(\gamma_k)})}
\label{outputstate}
\ee
For a given sequence of outcomes $\gamma_1,\dots,\gamma_k$, we refer to the state $\psi_k$ as the output state after a sample of an RMDC circuit. As we pointed out in the introduction we are interested in computing the purity of the output state $\psi_k$ in a given bipartition of the Hilbert space.

Consider  a bipartition $\mathcal{H}=\mathcal{H}_A\otimes\mathcal{H}_{B}$ of the system of qubits, let $\psi_A=\tr_B\psi$ be the marginal state on $\mathcal{H}_A$, then the purity of $\psi_{A}$ is defined as:
\be
\pur\psi_A:=\tr(\psi_{A}^{2})
\ee
Let $\tilde{S}_{A}\in\mathcal{B}(\mathcal{H}_{A}^{\otimes 2})$ be the swap operator on two copies of $\mathcal{H}_{A}$, then define $S_{A}\equiv\tilde{S}_{A}\otimes \bbbone_{B}^{\otimes 2}$, it is straightforward to verify
\be
\pur\psi_A=\tr(S_{A}\psi^{\otimes 2})
\label{purdecomposition}
\ee

For our scopes, the actual bipartition chosen is unimportant\cite{leone2021quantum} and thus, without loss of generality, we pick once and for all the bipartition corresponding to $d_{A}=d_{B}=\sqrt{d}$. Nevertheless the proofs of the main results are written in the general case and the case of $d_{A}=O(1)$ and $d_{B}=O(d)$ is analyzed in \ref{app:bla}.
We intend to compute the average purity over the ensemble of RMDC circuits by computing the average purity and its fluctuations over the Clifford operators $C_{\alpha}\in\mathcal{C}(d)$, $\alpha=1,\dots, k$, cfr. Eq. \eqref{outputstate}. We adopt a lighter notation for this average: $\aver{\pur\psi_{k,A}}_{C_{0},\dots,C_k\in\mathcal{C}(d)}\equiv \aver{\pur\psi_{k,A}}_{\mathcal{C}}$.

First, let us recall how the purity behaves for universal circuits and for Clifford circuits: if one takes $\psi_{U}=U\psi U^{\dag}$ to be the output state of an universal circuit $U\in\mathcal{U}(d)$, the values of the average purity and the fluctuations of purity are given by \cite{leone2021quantum}:
\ba
\aver{\pur \psi_{U,A}}_{U\in\mathcal{U}(d)}&=&\frac{2\sqrt{d}}{d+1}=\Omega(d^{-1/2})\label{avpurityhaar}\\
\Delta_{U\in\mathcal{U}(d)}\pur \psi_{U,A}&=&\frac{2(d-1)^2}{(d+1)^2(d+2)(d+3)}\nonumber\\&=&\Omega(d^{-2})
\ea
where we define:
\be 
\Delta_{U\in x}\pur\psi_{U,A}:=\aver{\pur^2\psi_{U,A}}_{U\in x}-\aver{\pur\psi_{U,A}}_{U\in x}^{2}
\ee
the ensemble purity fluctuations with respect to the ensemble of unitaries $x$. For $\psi_{C}=C\psi C^{\dag}$ being the output state of a Clifford circuit, one gets \cite{leone2021quantum}:
\ba
\hspace{-0.7cm}\aver{\pur \psi_{C,A}}_{C\in\mathcal{C}(d)}\hspace{-0.3cm}&=&\hspace{-0.3cm}\frac{2\sqrt{d}}{d+1}=\Omega(d^{-1/2})\\
\hspace{-0.7cm}\Delta_{C\in\mathcal{C}(d)}\pur \psi_{C,A}\hspace{-0.3cm}&=&\hspace{-0.3cm}\frac{(d-1)^2}{(d+1)^2(d+2)}=\Omega(d^{-1})\hspace{0.65cm}
\ea
While the average purity remains the same, the scaling of the purity fluctuations differs between a universal circuit and a Clifford one. As argued in \cite{leone2021quantum} a detection of the transition towards the universal behavior is given by the scaling in $d$ of the fluctuations of purity, from $\Omega(d^{-1})$ to $\Omega(d^{-2})$. To the aim of understanding how measurements can drive the transition towards the universal behavior of the purity fluctuations in RDMC circuits, let us introduce two families of basis: the Clifford basis and the $B_{\theta}$-basis:

\begin{definition}\label{def1}
Let $B_{c}:=\spann{\ket{0},\ket{1}}$ be the single qubit computational basis, then a Clifford basis $B_{\mathcal{C}}$ is:
\be
B_{\mathcal{C}}:=CB_{c}
\ee
where $C\in\mathcal{C}(2)$.
\end{definition}
\begin{definition}\label{thetabasis}
Let $\theta\in[0,\pi/2]$, then define the single qubit basis $B_{\theta}:=\spann{\ket{0}+e^{i\theta}\ket{1},\ket{0}-e^{i\theta}\ket{1}}$. 
\end{definition}
While Clifford bases are all the bases obtained from the computational one with the application of the single qubit Clifford group, the single qubit basis $B_{\theta}$ is obtained from the computational basis $B^{c}$ with the action of the Hadamard gate $H$ and the $K_{\theta}$-gate, namely $B_{\theta}=K_{\theta}HB_{c}$. 
The $K_{\theta}$-gate, defined as $K_{\theta}=\ket{0}\bra{0}+e^{i\theta}\ket{1}\bra{1}$ is a Clifford operator iff $\theta=0,\pi/2$.

In the next section, we prove two results: (i) $k=\Omega(n)$ single qubit one shot measurements in the $B_{\theta}$-basis, for $\theta\neq 0,\pi/2$, are both necessary and sufficient to drive the transition towards the universal behavior of the fluctuations of purity and (ii) there is no transition if the measurements are made in any Clifford basis.

\subsection{Main results}\label{sec:mainresult}
Here we present the main results of this paper: in Theorem \ref{th1} we compute the average purity for RMDC circuits and prove that it remains the same up to a negligible error, in Theorem \ref{th2} we compute the purity fluctuations for RMDC circuits.
\begin{theorem}\label{th1}
Let the initial state $\psi=\ket{0}\bra{0}^{\otimes n}$ and $d_A=d_B=\sqrt{d}$, then the average purity for RMDC circuits in Fig. \ref{scheme1} and measurements made in the basis $B_{\theta}$ reads:
\be
\aver{\pur\psi_{k,A}}_{\mathcal{C}}=\frac{2\sqrt{d}}{d+1}+\Omega(kd^{-3/2})
\ee
\end{theorem}
See \ref{th1proof} for the proof. The above result is sufficient to our scopes: although the average purity for RMDC circuits equals the universal value, see \eqref{avpurityhaar}, up to an error which becomes significant for $k=\Omega(d)$, we just need the number measurements $k$ to scale as $\Omega(\log d)$, as the next theorem shows, to get the purity fluctuations scaling as $\Omega(d^{-2})$ and thus to observe the transition from the Clifford behavior to the universal one. 
\begin{theorem}\label{th3}
Let the initial state $\psi=\ket{0}\bra{0}^{\otimes n}$ and $d_A=d_B=\sqrt{d}$, then the fluctuations of the purity for RMDC circuits and measurements made in the basis $B_{\theta}$ reads:
\be
\Delta_{\mathcal{C}}\pur\psi_{k,A}=\Omega\left(\left(\frac{7+\cos(4\theta)}{8}\right)^{2k}d^{-1}+ p d^{-2}\right)\label{mainresult}
\ee
where $p=\Omega(1)$. 
\end{theorem}
The proof is given in \ref{th3proof}.
\begin{remark}\label{remark1} The result of Theorem \ref{th2} does not depend on which of the projector of the basis $B_\theta$ we select at each layer $k$. Indeed $Z(\ket{0}+e^{i\theta}\ket{1})=\ket{0}-e^{i\theta}\ket{1}$ and trivially $Z\in\mathcal{C}(2)$. For the left/right invariance of the Haar measure, the average over the Clifford group is preserved. Thus we consider, without loss of generality, a given sequence of outcomes, namely $\gamma_{\alpha}=1$ for any $\alpha$, and drop the superscript on $P_{i}$ in the rest of the paper.
\end{remark}
\begin{corollary}
For any $\theta\neq 0,\frac{\pi}{2}$; iff $k=\Omega(\log d)$, then:
\be
\Delta_{\mathcal{C}}\pur\psi_{k,A}=\Omega(d^{-2})
\ee
\end{corollary}
\begin{proof} From Eq. \eqref{mainresult}, we have:
\be
\Omega\left(\left(\frac{7+\cos(4\theta)}{8}\right)^{2k}\right)=\Omega(d^{-1})\iff k=\Omega(\log d)
\ee
\end{proof}

In the next corollary we show how measurements in Clifford basis cannot drive any transition to universal behavior of the fluctuations of purity:

\begin{corollary}\label{th2}
For any $k$, the fluctuations of purity for RMDC circuits with measurements made in a Clifford basis read:
\be
\Delta_{\mathcal{C}}\pur\psi_{k,A}=\Omega(d^{-1})
\ee
\end{corollary}
\begin{proof}
Starting from Eq. \eqref{mainresult} for $B_{\pi/2}$, we have $\Delta_{C}\pur\psi_{k,A}=\Omega(d^{-1})$; this result holds for any Clifford basis $B_{\mathcal{C}}$ because of the left/right invariance of the Haar measure over the Clifford group\cite{collins2003moments,collins2006integration}.
\end{proof}

\begin{remark}
The minimum value of the scaling factor $7+\cos(4\theta)$ in Eq. \eqref{mainresult} is achieved for $\theta=\pi/4$, that is, for $B_{\theta}$ obtained from the computational basis applying the Hadamard-gate and the $T$-gate. Note that in \cite{leone2021quantum} we found a similar behavior for Clifford circuits interleaved by $T$ gates.(See Fig.\ref{fig:DensityPlot})
\end{remark}
\begin{figure}[h!]
    \centering
    \includegraphics[scale=0.33]{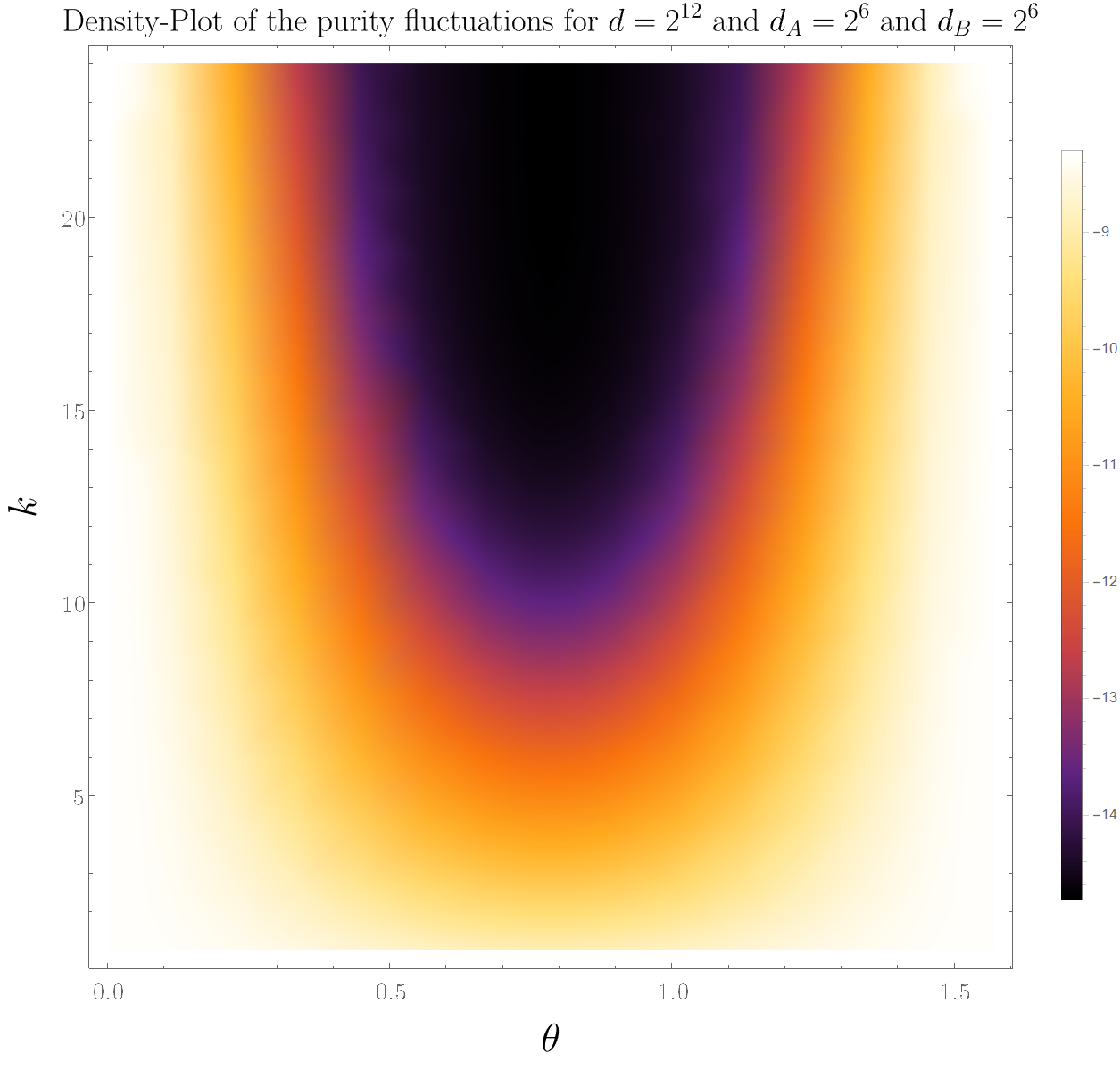}
    \caption{Density plot of the logarithm of the purity fluctuations (see Eq.\eqref{mainresult}) as function of the angle $\theta$ of the $B_\theta-$basis and the number of layers $k$, in the  case $d=2^{12}$ and $d_A=d_B=2^{6}$.}
    \label{fig:DensityPlot}
\end{figure}
\section{Repeated measurements}\label{sec:nonpost}
In previous sections we showed how one shot measurements in a Clifford circuit could drive the transition to an universal behavior. At this point, a natural question arises: what if we collect the outcome of a measurement many times? In the following, we show how collecting the outcomes of repeated measurements affects the evolution through the circuit. Consider a $n$-qubits reference state $\psi$ and let $\psi$ undergo a Clifford evolution $C_0$:
\be
\psi\mapsto \psi_0\equiv C_0\psi C_{0}^{\dag}
\ee
then applying the measurements in the basis $B_{i}=\spann{\ket{\phi_1}_i,\ket{\phi_2}_i}$ on a random qubit $i_{1}$ followed by another Clifford evolution $C_1$:
\be
\psi_0\mapsto \psi_{1}\equiv C_{1}\mathcal{D}_{B_i}(C_{0}\psi C_{0}^{\dag})C_{1}^{\dag}
\ee
recall that the collection of many outcomes of a measurement returns a mixed state obtained by the action of the dephasing superoperator $\mathcal{D}_{B_{i_1}}$\cite{Hamma2021coherence}, in the basis $B_{i_1}$, defined as:
\ba
\hspace{-0.7cm}\mathcal{D}_{B_i}(\cdot)\hspace{-0.3cm}&:=&\hspace{-0.3cm}(\ket{\phi_1}_i\bra{\phi_1}_i\otimes \bbbone^{(\bar{i})})(\cdot)(\ket{\phi_1}_i\bra{\phi_1}_i\otimes \bbbone^{(\bar{i})})\nonumber\\\hspace{-0.7cm}\hspace{-0.3cm}&+&\hspace{-0.3cm}(\ket{\phi_2}_i\bra{\phi_2}_i\otimes \bbbone^{(\bar{i})})(\cdot)(\ket{\phi_2}_i\bra{\phi_2}_i\otimes \bbbone^{(\bar{i})})\hspace{0.6cm}
\label{db}
\ea
Iterating the process $k$ times, the output state $\psi_k$ reads:
\be
\psi_{k}\equiv C_{k}\mathcal{D}_{B_{i_k}}(\cdots C_{1}\mathcal{D}_{B_{i_1}}(C_0\psi C_{0}^{\dag})C_{1}^{\dag}\cdots)C_{k}^{\dag}
\ee
we thus have linear Completely Positive Trace Preserving (CPTP) maps $\mathcal{D}_{B_{i_\alpha}}$ interleaved with global Clifford operators $C_{\alpha}$, for $\alpha=1,\dots, k$. To compare this protocol with RMDC circuits discussed in Sec. \ref{sec:setup}, we compute the average purity and its fluctuations for the output state $\psi_k$ by averaging over all the Clifford operators. Let us now define:
\begin{definition}
Let $\mathcal{M}$ be a linear CPTP map, $\mathcal{O}_2\in\mathcal{B}(\mathcal{H}^{\otimes 2})$ and $\mathcal{O}_4\in\mathcal{B}(\mathcal{H}^{\otimes 4})$, then the $(\mathcal{M},\mathcal{C}_k)$-fold channels of order $2$ and $4$ are defined:
\ba
\hspace{-0.7cm}\Phi^{(2)}_{(\mathcal{M},\mathcal{C}_k)}(\mathcal{O}_2)\hspace{-0.3cm}&:=&\hspace{-0.3cm}\aver{C_{k}^{\otimes 2}\mathcal{M}(\dots C_{1}^{\otimes2}\mathcal{M}(C_{0}^{\otimes 2}\mathcal{O}_2C_{0}^{\dag\otimes 2}))}_{C_{1},\dots, C_k}\\
\hspace{-0.7cm}\Phi^{(4)}_{(\mathcal{M},\mathcal{C}_k)}(\mathcal{O}_4)\hspace{-0.3cm}&:=&\hspace{-0.3cm}\aver{C_{k}^{\otimes 4}\mathcal{M}(\dots C_{1}^{\otimes 4}\mathcal{M}(C_{0}^{\otimes 2}\mathcal{O}_4C_{0}^{\dag\otimes 4}))}_{C_{1},\dots, C_k}
\ea
where $C_{1},\dots, C_{k}\in\mathcal{C}(d)$ are Clifford operator and the average is taken according to the Haar measure over groups.
\end{definition}
Thanks to the left/right invariance of the Haar measure over the Clifford group, the average purity and its fluctuations do not depend on which qubit we measure at each step; thus without loss of generality we can take measurements always on the same qubit $i_1$, $\mathcal{D}_{B_{i_\alpha}}=\mathcal{D}_{B_{i_1}}$ for any $\alpha$. Therefore, we can express the average purity and its fluctuations in terms of $(\mathcal{D}_{B},\mathcal{C}_k)$-fold channels:
\ba
\aver{\pur \psi_{k,A}}_{\mathcal{C}}&\equiv&\tr(S_{A}\Phi^{(2)}_{(\mathcal{D}_B,\mathcal{C}_k)}(\psi^{\otimes 2}))\label{avpurkfoldchannel}\\
\Delta_{D_{B},\mathcal{C}_k}(\mathcal{O}_2)&\equiv&\tr(T^{(A)}_{(12)(34)}\Phi_{(\mathcal{D}_B,\mathcal{C}_k)}^{(4)}(\psi^{\otimes 4}))\nonumber\\&-&\tr(S_{A}\Phi^{(2)}_{(\mathcal{D}_B,\mathcal{C}_k)}(\psi^{\otimes 2}))^{2}\label{purfluckfoldchannel}
\ea
where we adopted the lighter notation $\mathcal{D}_B\equiv\mathcal{D}_{B_{i_1}}$.
Before specializing our calculations for the dephasing superoperator $\mathcal{D}_B$, we can give a general theorem regarding the fold channels of order $2$ and $4$:
\begin{theorem}\label{th4}
Let $\mathcal{M}$ a linear CPTP quantum map. Let $\mathcal{O}_2\in\mathcal{B}(\mathcal{H}^{\otimes 2})$ and $\mathcal{O}_4\in\mathcal{B}(\mathcal{H}^{\otimes 4})$ two linear operators on $\mathcal{H}^{\otimes 2}$ and $\mathcal{H}^{\otimes 4}$ respectively. Then the $(\mathcal{M},\mathcal{C}_k)$-fold channels $\Phi^{(2)}_{\mathcal{M},\mathcal{C}_k}(\mathcal{O}_2)$ and $\Phi^{(4)}_{\mathcal{M},\mathcal{C}_k}(\mathcal{O}_4)$, of order $2$ and $4$, read:
\be
\Phi_{\mathcal{M},\mathcal{C}_k}^{(2)}(\mathcal{O}_{2})=\sum_{\rho\in  S_2}a_{\rho}^{(k)}(\mathcal{O}_{2})T_{\rho}
\ee
where $a_{\rho}^{(k)}(\mathcal{O}_2)\equiv\sum_{\sigma\in  S_2}(\Xi^{k})_{\rho\sigma}a_{\rho}^{(0)}(\mathcal{O}_{2})$, $\Xi_{\rho\sigma}\equiv\sum_{\kappa\in S_2}W_{\rho\kappa}\tr(\mathcal{M}(T_{\kappa})T_{\sigma})$ and $\Xi^{k}$ is the $k$-th matrix power of  $\Xi$, while $a_{\rho}^{(2)}(\mathcal{O})=\sum_{\kappa\in S_2}W_{\rho\kappa}\tr(\mathcal{O}T_{\kappa})$, with $W_{\rho\kappa}$ the Weingarten functions defined in \eqref{weingartenhaar} and
\be
\Phi_{\mathcal{M},\mathcal{C}_k}^{(4)}(\mathcal{O}_{4})=\sum_{\rho\in S_4}(c^{(k)}_{\rho}(\mathcal{O}_4)Q+b^{(k)}_{\rho}(\mathcal{O}_4))T_{\rho}
\ee
where the coefficients $c_{\rho}^{(k)}(\mathcal{O}_4)$ and $b_{\rho}^{(k)}(\mathcal{O}_4)$ obey to the following recurrence relations:
\ba
\hspace{-0.7cm}c_{\rho}^{(k)}(\mathcal{O}_4)\hspace{-0.3cm}&\equiv&\hspace{-0.3cm}\sum_{\sigma\in S_4}(M_{\rho\sigma}c_{\sigma}^{(k-1)}(\mathcal{O}_4)+N_{\rho\sigma}b_{\sigma}^{(k-1)}(\mathcal{O}_4))\hspace{0.6cm}\\
\hspace{-0.7cm}b_{\rho}^{(k)}(\mathcal{O}_4)\hspace{-0.3cm}&\equiv&\hspace{-0.3cm}\sum_{\sigma\in S_4}(O_{\rho\sigma}c_{\sigma}^{(k-1)}(\mathcal{O}_4)+P_{\rho\sigma}b_{\sigma}^{(k-1)}(\mathcal{O}_4))\label{ckbkth3}
\ea
with the following initial conditions:
\ba
\hspace{-0.7cm}c_{\rho}^{(0)}\hspace{-0.3cm}&=&\hspace{-0.3cm}\sum_{\sigma\in S_4}W^{+}_{\rho\sigma}\tr(\mathcal{O}_4QT_{\sigma})-W^{-}_{\rho\sigma}\tr(\mathcal{O}_4Q^{\perp}T_{\sigma})\\
\hspace{-0.7cm}b_{\rho}^{(0)}\hspace{-0.3cm}&=&\hspace{-0.3cm}\sum_{\sigma\in S_4}W^{-}_{\rho\sigma}\tr(\mathcal{O}_4Q^{\perp}T_{\sigma})
\label{coefficientsinitialth}
\ea
where $W_{\rho\sigma}^{\pm}$ are the generalized Weingarten functions defined in \eqref{generalizedWeingarten} and:
\ba
\hspace{-0.7cm}M_{\kappa\rho}&\equiv&\sum_{\sigma\in S_4}(W^{+}_{\sigma\kappa}\tr(\mathcal{M}^{\otimes 4}(QT_{\rho})QT_{\sigma})\nonumber\\&-&W^{-}_{\sigma\kappa}\tr(\mathcal{M}^{\otimes 4}(QT_{\rho})Q^{\perp}T_{\sigma}))\nonumber\\
\hspace{-0.7cm}N_{\kappa\rho}&\equiv&\sum_{\sigma\in S_4}(W^{+}_{\sigma\kappa}\tr(\mathcal{M}^{\otimes 4}(T_{\rho})QT_{\sigma}))\nonumber\\&-&W^{-}_{\sigma\kappa}\tr(\mathcal{M}^{\otimes 4}(T_{\rho})Q^{\perp}T_{\sigma})\nonumber\\
\hspace{-0.7cm}O_{\kappa\rho}&\equiv&\sum_{\sigma\in S_4}W^{-}_{\sigma\kappa}\tr(\mathcal{M}^{\otimes 4}(QT_{\rho})Q^{\perp}T_{\sigma})\\
\hspace{-0.7cm}P_{\kappa\rho}&\equiv&\sum_{\sigma\in S_4}W^{-}_{\sigma\kappa}\tr(\mathcal{M}^{\otimes 4}(T_{\rho})Q^{\perp}T_{\sigma})\nonumber\label{24by24matrices}
\ea
\end{theorem}
The proof is given in  \ref{th4proof}. In the following section we use this general theorem to answer the question posed at the beginning and conclude that collecting the outcomes of repeated measurements in a Clifford circuit behave very differently with respect to one shot ones: (i) collecting the outcomes of repeated measurements does not drive the transition towards the universal behavior of purity fluctuations and (ii) it does not discriminate between Clifford and non Clifford basis, as one shot measurements do.
\subsection{Subsystem purity}
Equipped with Theorem \ref{th4}, we can compute the average purity and its fluctuations for the mixed state obtained by the collection of many outcomes and thus use the general result of Theorem \ref{th4} for $\mathcal{M}=\mathcal{D}_B$, defined in Eq. \eqref{db}. In Proposition \ref{corollarydbstates2and4} we compute the $(\mathcal{D}_{B},\mathcal{C}_k)$-fold channels for $\psi^{\otimes 2}$ and $\psi^{\otimes 4}$ while in the subsequent corollaries we compute the average purity and its fluctuations according to Eqs. \eqref{avpurkfoldchannel} and \eqref{purfluckfoldchannel}.

Let us consider single qubit measurements in the basis $B_{\theta}\equiv \spann{\ket{0}\pm e^{i\theta}\ket{1}}$, defined in Definition \ref{thetabasis}, with the following lighter notation: $\mathcal{D}_{B_{\theta}}\equiv \mathcal{D}_{\theta}$.

\begin{prop}\label{corollarydbstates2and4}
The $(\mathcal{D}_{\theta},\mathcal{C}_k)$-fold channels of $\psi^{\otimes 2}$ and $\psi^{\otimes 4}$ read:
\be
\Phi^{(2)}_{\mathcal{D}_{\theta},C_k}(\psi^{\otimes 2})=\sum_{\rho\in  S_2}a_{\rho}^{(k)}(\psi^{\otimes 2})T_{\rho}
\label{psikfoldchannel2}
\ee
where the coefficients for any $\theta$ are given by:
\ba
a_{e}^{(k)}(\psi^{\otimes 2})&=&\frac{1}{d^2}-\frac{h^k}{d^2(d+1)}\nonumber\\
a_{(12)}^{(k)}(\psi^{\otimes 2})&=&\frac{h^k}{d^2(d+1)}
\label{psikfoldchannel2coefficients}
\ea
and 
\be
\Phi^{(4)}_{\mathcal{D}_{\theta},C_k}(\psi^{\otimes 4})=\sum_{\rho\in  S_4}[c_{\rho}^{(k)}(\psi^{\otimes 4})Q+b_{\rho}^{(k)}(\psi^{\otimes 4})]T_{\rho}
\label{psifoldchannel4}
\ee
where the coefficients $c_{\rho}^{(k)}, b_{\rho}^{(k)}$ for $\theta=\pi/2$ are explicitly calculated in  \ref{sec:4foldchannel}.

Moreover for $k\rightarrow\infty$ we obtain:
\ba
\lim_{k\rightarrow\infty}\Phi^{(2)}_{\mathcal{D}_{\theta},C_{k}}(\psi^{\otimes 4})&=&\frac{\bbbone^{\otimes 2}}{d^{2}}\\
\lim_{k\rightarrow\infty}\Phi^{(4)}_{\mathcal{D}_{\pi/2},C_{k}}(\psi^{\otimes 4})&=&\frac{\bbbone^{\otimes 4}}{d^{4}}
\ea
\end{prop}
The proof can be found in  \ref{propinfinitestatesproof}.
\begin{corollary}\label{avpurdbpi2}
The average purity for any $\theta$ and $d_A=d_B=\sqrt{d}$ reads:
\be
\aver{\pur \psi_{k,A}}_{\mathcal{D}_\theta,\mathcal{C}_k}=\frac{f^{k}(d-1)+(d+1)}{\sqrt{d}(d+1)}\label{purdb}
\ee
where $f=\frac{(d^2-8)}{8(d^2-1)}$ and for $k\rightarrow\infty$ we have $\aver{\pur \psi_{k,A}}_{\mathcal{D}_\theta,\mathcal{C}_k}\rightarrow d^{-1/2}$.
\end{corollary}
\begin{proof} From Eq. \eqref{avpurkfoldchannel} the proof is a straightforward application of Proposition \ref{corollarydbstates2and4}. \end{proof}
\begin{corollary}\label{flucpurdbpi2}
The fluctuations of purity for $\theta=\pi/2$ and $d_A=d_B=\sqrt{d}$ reads:
\ba
\hspace{-1cm}\Delta_{\mathcal{D}_{\pi/2},\mathcal{C}_k}\pur \psi_{A}\hspace{-0.3cm}&=&\hspace{-0.3cm}\frac{(d-1)(d-4)}{d(d+1)(d+2)}g^{k}+\frac{3(d-1)}{d(d+1)}h^{k}\nonumber\\
\hspace{-1cm}&-&\hspace{-1cm}\frac{2(d-1)}{d(d+1)}f^{k}-\frac{(d-1)^{2}}{d(d+1)^{2}}f^{2k}
\label{purflucpi2cor}
\ea
where:
\be
f=\frac{(d^2-8)}{8(d^2-1)}, \quad g=\frac{(d^2-4)}{4(d^2-1)}, \quad h=\frac{(d^2-2)}{2(d^2-1)}
\ee
while for $k\rightarrow\infty$ we have $\Delta_{\mathcal{D}_{\pi/2},\mathcal{C}_k}\pur \psi_{A}\rightarrow 0$.
\end{corollary}
\begin{proof} From Eq. \eqref{purfluckfoldchannel}, the proof is a straightforward application of Proposition \ref{corollarydbstates2and4}. \end{proof}

In the next Proposition, we prove that after $\Omega(n)$ steps the purity reaches its minimum value $d^{-1/2}$, and its fluctuations can be set arbitrary small, meaning that the state $\psi$ is getting more and more mixed towards the completely mixed state.
\begin{prop}\label{prop2}
For $k=\alpha n$ with $\alpha>0$, the average purity for any $\theta$ equals $d^{-1/2}$ up to an error $\Omega(d^{-\alpha \log_2 f})$, while the purity fluctuations for $\theta=\pi/2$ can be set to be $\Omega(d^{-\alpha \log_2 h})$.
\end{prop}
\begin{proof} Set $k=\alpha n$, then from Eq. \eqref{purdb}, since $0<f<1$, we have:
\be
\aver{\pur\psi_{A}}_{(\mathcal{D}_\theta,C_{\Omega(n)})}=\frac{1}{\sqrt{d}}+\Omega(d^{-\alpha\log_2 f}),
\ee
and from Eq.  \eqref{purflucpi2cor}, and for any $k\ge0$ we have the following bound:
\be
\Delta_{(\mathcal{D}_{\pi/2},\mathcal{C}_k)}\pur\psi_A <h^{k}
\ee
where $h$ is defined in Eq. \eqref{functionsfgh}. Since $0<h<1$, we have that if $k>\alpha n$, where $\alpha>0$
\be
\Delta_{(\mathcal{D}_{\pi/2},\mathcal{C}_k)}\pur\psi_A< d^{-\alpha\log_2h}
\ee
this concludes the proof. \end{proof}

In Proposition \ref{corollarydbstates2and4} we computed the fold channel of order $4$ just for $\theta=\pi/2$ and in Proposition \ref{prop2} we showed that the purity fluctuations can be set arbitrary small with $k=\Omega(n)$ iterations; these two striking results do not depend on the fact that $B_{\pi/2}$ is a Clifford basis (cfr. Definition \ref{thetabasis}). In the next Lemma we show that, for $\theta=\pi/4$, both the asymptotic value of the $(\mathcal{D}_{\pi/4},\mathcal{C}_k)$-fold channel $\Phi^{(4)}_{(\mathcal{D}_{\pi/4},\mathcal{C}_k)}(\psi^{\otimes 4})$ and the convergence rate coincide with the ones for $\theta=\pi/2$. We thus conclude that, contrary to the case of one shot measurements, here $\pi/2$ is not a special, fine-tuned case: repeated measurements do not discriminate between Clifford and non Clifford basis, giving in both cases the completely mixed state $\psi\propto\bbbone$.
\begin{lemma}\label{lemmainfinitestates}
For $\theta=\pi/4$, the following results hold:
\be
\lim_{k\rightarrow\infty}\Phi_{(\mathcal{D}_{\pi/4},\mathcal{C}_k)}^{(4)}(\psi^{\otimes 4})=\frac{\bbbone^{\otimes 4}}{d^4}
\ee
and for $k=\alpha n$ with $\alpha>0$ the fluctuations of the purity can be set to be $\Omega(d^{-\alpha\log_2 h})$.
\end{lemma}
The proof can be found in  \ref{lemmainfinitestatesproof}.
\section*{Conclusions and Outlook}
In this paper, we showed that $\Omega(n)$ one shot measurements in a non Clifford basis in a $n$-qubit random Clifford circuit drive the transition to the third class of entanglement complexity, defined by the adherence to the universal value of the ensemble purity fluctuations. We also analyze the case of multiple measurements, showing that it is not suitable to drive the complexity transition: after $\Omega(n)$ measurements the output mixed state $\psi$ looks like the completely mixed state.

In perspective, there are several open questions. One could generalize these results by considering multiple qubit measurements: on the one hand is clear that measuring all the qubit at once cannot drive any transition because this would result in completely factorizing the state;  on the other hand, increasing the density of measurements, it would be interesting to investigate if there is a threshold from which there are no complexity transitions. A related question is that of whether $\Omega(n)$ measurements are enough to reproduce the full ESS. More generally, we find it important to address the question of the classification of entanglement complexity classes. Finally, given a finite number of non-Clifford resources, one could ask what is the most efficient placement of non-Clifford gates and measurements to achieve the desired universal features.

\section*{Acknowledgments} We acknowledge support from  NSF award number 2014000. The authors thank Jacob Hauser for important and enlightening discussions and comments.
\appendix
\section{Mathematical preliminaries}\label{app:mathpre}
\subsection{Haar measure over groups}
Here we give a review on the Haar average over the full unitary group and the Clifford group.
\subsubsection{Unitary group average}
 Given $\mathcal{O}\in\mathcal{B}(\mathcal{H}^{\otimes t})$ a bounded operator on $t$-copies of $\mathcal{H}$ the Haar average reads
\ba
\aver{\mathcal{O}_U}_{U\in\haar}&:=&\int_{\mathcal{U}(d)}\de U U^{\dag\otimes t}\mathcal{O}U^{\otimes t}\nonumber\\&=&\sum_{\rho,\sigma\in S_{t}}W_{\rho\sigma}\tr(\mathcal{O}T_{\sigma})T_{\rho}
\label{haarformulageneral}
\ea
where $T_{\rho}$ is the permutation operator standing for the permutation $\rho \in S_{t}$, the symmetric group of order $t!$ and $W_{\rho\sigma}$ the Weingarten functions defined as
\be
W_{\rho\sigma}\equiv\sum_{\lambda\vdash t}\frac{d_{\lambda}^{2}}{(t!)^2}\frac{\chi^{\lambda}(\rho\sigma)}{D_{\lambda}}\label{weingartenhaar}
\ee
where $\lambda$ label the irreducible representations of $ S_t$ with dimensions $d_\lambda$, $D_{\lambda}=\tr(\rho_{t}^{\lambda})$, with $\Pi_4^{\lambda}$ the projectors on the irreducible representations of $ S_t$, and $\chi^{\lambda}(\rho\sigma)$, are the characters of the irreducible representations $\lambda$ of $ S_t$.
More details about the Haar average  in \cite{collins2003moments,collins2006integration}.
\subsubsection{Clifford group average}\label{Cliffhaar}
Given $\mathcal{O}\in\mathcal{B}(\mathcal{H}^{\otimes 4})$, the integration formula for the Clifford group reads
\ba
\aver{\mathcal{O}_{C}}_{\mathcal{C}}&:=&\int_{\mathcal{C}(d)}\de C C^{\dag\otimes 4}\mathcal{O}C^{\otimes 4}\nonumber\\&=&\sum_{\rho,\sigma\in S_{4}}W^{+}_{\rho\sigma}\tr(\mathcal{O}QT_\rho)QT_{\sigma}\\&+&W^{-}_{\rho\sigma}\tr(\mathcal{O}Q^{\perp}T_{\rho})Q^{\perp}T_{\sigma}\nonumber
\label{cliffordweing}
\ea
where $Q=\frac{1}{d^2}\sum_{P\in\mathcal{P}(d)}P^{\otimes 4}$ and $Q^{\perp}=\bbbone^{\otimes 4}-Q$, $T_{\sigma}$ are permutation operators standing for the permutation $\sigma\in  S_4$, while $W^{\pm}_{\rho\sigma}$ are the generalized Weingarten functions, defined as
\be
W^{\pm}_{\rho\sigma}\equiv\sum_{\substack{\lambda\vdash 4\\D^{\pm}_{\lambda}\neq 0}}\frac{d_{\lambda}^2}{(4!)^2}\frac{\chi^{\lambda}(\rho\sigma)}{D^{\pm}_\lambda}
\label{generalizedWeingarten}
\ee
as for the Haar average on the unitary group $\lambda$ labels the irreducible representations of the symmetric group $ S_4$, and $\chi^{\lambda}(\rho\sigma)$ are the characters of $ S_4$, $d_\lambda$ is the dimension of the irreducible representation $\lambda$, while $D_{\lambda}^{+}=\tr(Q\Pi^{\lambda}_4)$ and $D_{\lambda}^{-}=\tr(Q^{\perp}\Pi^{\lambda}_4)$. 
More details in \cite{leone2021quantum,roth2018recovering}
\subsection{Unitary $t$-design}
In this section, we provide the formal definition of unitary $t$-design, more details can be in \cite{Gross2007unitary,webb2016clifford,zhu2017multiqubit}. 
Let us consider a system of $n$-qubit with Hilbert space $\mathcal{H}=\bigotimes_{i=1}^{n}\mathbb{C}_{j}^{2}$. Given an ensemble of unitaries with a fixed probability distribution $\mathcal{E}=\{p_{i},U_{i}\}$, $\mathcal{E}$ is a unitary $t-$design iff 
\be 
\sum_{i}p_{i}U_{i}^{\otimes t}\rho U_{i}^{\dag \otimes t} = \int_{\mathcal{U}(d)}\de U U^{\otimes t} \rho U^{\dag \otimes t} 
\ee 
for all the quantum states $\rho\in \mathcal{H}^{\otimes  t}$. If this condition is satisfied the ensemble $\mathcal{E}$ is able to reproduce the statistics of at least $t$ moments of the uniform distribution over the unitary group.
\subsection{Average of a ratio\label{ratioaverage}}
Let $x,y$ two stochastic dependent variables. Here we want give an approximation of two quantities:
\be
r(x,y):=\aver{\frac{x}{y}}, \quad r_{2}(x,y):=\aver{\left(\frac{x}{y}-\aver{\frac{x}{y}}\right)^{2}}
\ee
and bound the error. First $r(x,y)$. Let us Taylor-expand $\frac{1}{y}$ around $\aver{y}$:
\be
\frac{1}{y}=\frac{1}{\aver{y}}+\Omega\left(-\frac{1}{\aver{y}^2}(y-\aver{y})\right)
\ee
then:
\ba
r(x,y)\hspace{-0.3cm}&=&\hspace{-0.3cm}\aver{\frac{x}{\aver{y}}}+\Omega\left(\frac{x}{\aver{y}^2}(y-\aver{y})\right)\\\hspace{-0.3cm}&=&\hspace{-0.3cm}\frac{\aver{x}}{\aver{y}}+\Omega\parent{\frac{\aver{(x-\aver{x})(y-\aver{y})}}{\aver{y}^2}}\nonumber
\label{rxy}
\ea
For $r_{2}(x,y)$
\ba
r_{2}(x,y)&=&\aver{\left(\frac{x}{y}-\frac{\aver{x}}{\aver{y}}+\epsilon\right)^{2}}\nonumber\\
&=&\aver{\left(\frac{x}{y}-\frac{\aver{x}}{\aver{y}}\right)^{2}}+\epsilon^2\\&+&2\aver{\left(\frac{x}{y}-\frac{\aver{x}}{\aver{y}}\right)}\epsilon\nonumber
\ea
From Eq. \eqref{rxy} we have
\be
\aver{\left(\frac{x}{y}-\frac{\aver{x}}{\aver{y}}\right)}=\epsilon
\ee
Thus:
\be
r_{2}(x,y)=\aver{\left(\frac{x}{y}-\frac{\aver{x}}{\aver{y}}\right)^{2}}+3\epsilon^2
\ee
Just like above, let us expand $\frac{x}{y}$ around the mean $\frac{\aver{x}}{\aver{y}}$
\be
\frac{x}{y}=\frac{\aver{x}}{\aver{y}}+\Omega\left(\frac{1}{\aver{y}}(x-\aver{x})-\frac{\aver{x}}{\aver{y}^2}(y-\aver{y})\right)
\ee
then
\be
\aver{\left(\frac{x}{y}-\frac{\aver{x}}{\aver{y}}\right)^{2}}=\aver{\Omega\left(\frac{1}{\aver{y}}(x-\aver{x})-\frac{\aver{x}}{\aver{y}^2}(y-\aver{y})\right)^{2}}
\ee
After some algebra we find that:
\ba
\hspace{-0.7cm}r_{2}(x,y)&=&\Omega\left(\frac{\aver{(x-\aver{x})^2}}{\aver{y}^2}+\frac{\aver{(y-\aver{y})^2}\aver{x}^2}{\aver{y}^{4}}\right.\nonumber\hspace{0.7cm}\\&-&\left.2\frac{\aver{(x-\aver{x})(y-\aver{y})}\aver{x}}{\aver{y}^{3}}\right)+3\epsilon^2
\label{r2xy}
\ea
\section{Main theorems}
\subsection{Proof of Theorem \ref{th1}\label{th1proof}}
In order to prove Theorem \ref{th1} we need to compute the average purity over RMDC circuits. Taking the average over the Clifford operators $C_{\alpha}$, $\alpha=1,\dots, k$ and from Eq. \eqref{purdecomposition} we have:
\be
\aver{\pur\psi_{k,A}}_{\mathcal{C}}=\tr(S_{A}\aver{\psi^{\otimes 2}_{k}}_{\mathcal{C}})
\ee
First, we define the non-normalized output state:
\be
\hat{\psi}_{k}:=C_{k}P_{k}\cdots C_{1}P_{1}C_{0}\psi C_{0}^{\dag}P_{1}C_{1}^{\dag}\cdots P_{k}C_{k}^{\dag}
\label{psik}
\ee
The average purity then
\be
\aver{\pur\psi_{k,A}}_{\mathcal{C}}=\aver{\frac{\pur \hat{\psi}_{k,A}}{N_k}}_{\mathcal{C}}
\ee
where $N_{k}=\tr(\hat{\psi}^{\otimes 2})$.
We make the following approximation:
\be
\aver{\pur\psi_{k,A}}_{\mathcal{C}}= \frac{\aver{\pur \hat{\psi}_{k,A}}_{\mathcal{C}}}{\aver{N_k}_{\mathcal{C}}}+\epsilon
\label{approx}
\ee
and compute the error $\epsilon$:
\be
\epsilon=\Omega\parent{\frac{\aver{\pur\hat{\psi}_{k,A}N_k}_{\mathcal{C}}-\aver{\pur\hat{\psi}_{k,A}}_{\mathcal{C}}\aver{N_k}_{\mathcal{C}}}{\aver{N_k}_{\mathcal{C}}^{2}}}
\label{epsil}
\ee
see  \ref{ratioaverage} for the derivation of the error $\epsilon$. We have defined $\Delta_{\mathcal{C}}N_k\equiv \aver{N_{k}^{2}}_{\mathcal{C}}-\aver{N_k}_{\mathcal{C}}^{2}$.
First, from Eqs. \eqref{nk} and  \eqref{purpsia}
we obtain:
\be
\frac{\aver{\pur \hat{\psi}_{k,A}}_{\mathcal{C}}}{\aver{N_k}_{\mathcal{C}}}=\frac{2\sqrt{d}}{d+1}
\ee
Then from Eqs. \eqref{nk},  \eqref{purpsia}, \eqref{deltan}, \eqref{covnkpur} we compute the error for the basis $B_\theta$, using Eq.\eqref{coeffbasistheta}
\ba
\hspace{-1cm}\epsilon&=& 2 \csc ^2(2 \theta )\left(\frac{29k + 24-24 \left(\frac{(\cos (4 \theta )+7)}{8}\right)^{2k}}{(15+\cos (4 \theta ))}\right.\nonumber\\
&-&\left.\frac{4 \cos (4 \theta ) \left(7 k-2 \left(1-\left(\frac{(\cos (4 \theta )+7)}{8}\right)^{2k}\right)\right)}{(15+\cos (4 \theta ))}\right.\nonumber\hspace{-0.2cm}\\
&-&\left.\frac{ k \cos (8 \theta )}{(15+\cos (4 \theta ))}\right)\frac{1}{d^\frac{3}{2}}
\ea
Taking the asymptotic behavior for large $k$, we get $\epsilon=\Omega(kd^{-3/2})$. This concludes the proof.
\qed
\subsection{Proof of Theorem \ref{th3}\label{th3proof}}
From the definition of purity fluctuations $\Delta_{\mathcal{C}}\pur\psi_{k,A}:=\aver{\pur^2\psi_{k,A}}_{\mathcal{C}}-\aver{\pur\psi_{k,A}}_{\mathcal{C}}^{2}$ and from Eq. \eqref{r2xy} in  \ref{ratioaverage}, we have:
\ba
\hspace{-0.7cm}\Delta_{\mathcal{C}}\pur\psi_{k,A}\hspace{-0.3cm}&=&\hspace{-0.3cm}\Omega\left(\frac{\Delta_{\mathcal{C}}\pur\hat{\psi}_{k,A}}{\aver{N_k}_{\mathcal{C}}^2}
+\frac{\Delta_{\mathcal{C}}N_k \aver{\pur\hat{\psi}_{k,A}}_{\mathcal{C}}}{\aver{N_k}_{\mathcal{C}}^{4}}\right.\nonumber\\
\hspace{-0.7cm}\hspace{-0.3cm}&-&\hspace{-0.3cm}2\frac{\aver{\pur\hat{\psi}_{k,A} N_k}_{\mathcal{C}}\aver{\pur\hat{\psi}_{k,A}}_{\mathcal{C}}}{\aver{N_k}_{\mathcal{C}}^{3}}\\\hspace{-0.3cm}&+&\hspace{-0.3cm}\left.2\frac{\aver{\pur\hat{\psi}_{k,A}}_{\mathcal{C}}\aver{N_k}_{\mathcal{C}}\aver{\pur\hat{\psi}_{k,A}}_{\mathcal{C}}}{\aver{N_k}_{\mathcal{C}}^{3}}\right)\nonumber
\ea
It is sufficient to plug Eqs.\eqref{nk}, \eqref{purpsia}, \eqref{deltan}, \eqref{deltapur},  \eqref{covnkpur} and \eqref{coeffbasistheta} to obtain the desired result.\qed

\subsection{Calculations for the non-normalized output state $\hat{\psi}_k$ of RMDC circuits}\label{sec:nonnorm}
In this section, we develop the calculations regarding the non-normalized output state - labeled as $\hat{\psi}$ - obtained after $k$ one shot measurements of the initial state $\psi$.
For the aim of our calculations we need to estimate the averages of $\hat{\psi}^{\otimes 2}$ and $\hat{\psi}^{\otimes 4}$ and their applications to the subsystem purity, given a bipartition $\mathcal{H}=\mathcal{H}_A\otimes \mathcal{H}_B$, with $d_A=d_B=\sqrt{d}$. Let us set a notation we use throughout these proofs: let $s=2,4$, we denote as $A^{(i)}\in\mathcal{B}(\mathbb{C}_{i}^{2\otimes s})$, an operator $A$ whose support is on $s$ copies of the $i$-th qubit Hilbert space $\mathbb{C}_{i}^{2}$, while $B^{(\bar{i})}\in\mathcal{B}(\mathcal{H}^{(\bar{i})\otimes s})$ an operator $B$ whose support is on $s$ copies of $\mathcal{H}^{(\bar{i})}\equiv\bigotimes_{j\neq i}\mathbb{C}^{2}_{j}$.
\subsubsection{Calculations for $\aver{\hat{\psi}_{k}^{\otimes 2}}_{\mathcal{C}}$}
Let $B:=\spann{\ket{\phi_1}_i,\ket{\phi_2}_i}$ be a single qubit basis, then let $\tilde{P}^{(i)}$ be either $\ket{\phi_1}_i\bra{\phi_1}_i$ or $\ket{\phi_2}_i\bra{\phi_2}_i$ and $P_i\equiv\tilde{P}^{(i)}\otimes \bbbone^{(\bar{i})}$. The first Clifford average gives us 
\be
\aver{C_{0}^{\otimes 2}\hat{\psi}^{\otimes 2}C_{0}^{\dag\otimes 2}}_{C_0}\equiv\aver{\hat{\psi}_0^{\otimes 2}}_{\mathcal{C}}=\frac{\Pi_2}{\tr \Pi_2}\nonumber
\ee 
where $\Pi_2\equiv\frac{1}{2}(\bbbone^{\otimes 2}+S)$, with $S$ the swap operator between two copies of $\mathcal{H}$. Applying the first projector $P_{i_1}$ on the $i_1$-th qubit, we obtain 
\ba
\hspace{-1cm}\aver{\hat{\psi}_0^{\otimes 2}}_{\mathcal{C}}&\rightarrow& P_{i_1}^{\otimes 2}\aver{\hat{\psi}_0^{\otimes 2}}_{\mathcal{C}}P_{i_1}^{\otimes 2}\\& =& \frac{1}{\tr\Pi_2}P_{i_1}^{\otimes 2}\sum_{\rho\in S_2}T_{\rho}^{(i_1)}\otimes T_{\rho}^{(\overline{i_1})}P_{i_1}^{\otimes 2}\nonumber\\\hspace{-1cm}&=&\frac{\tilde{P}^{(i_1)\otimes 2}\otimes \Pi_2^{(\overline{i_1})}}{\tr{\Pi_2}}\nonumber
\ea
where we used the fact that $T_{\rho}=T_{\rho}^{(i_1)}\otimes T_{\rho}^{(\overline{i_1})}$ for any $\rho\in  S_2$ and $\Pi_2^{(\overline{i_1})}=\frac{1}{2}\parent{\bbbone^{\otimes 2}+S^{(\overline{i_1})}}$. Taking another Clifford average
\ba
\hspace{-1cm}\aver{C_{1}^{\otimes 2}P_{i_1}^{\otimes 2}\aver{\hat{\psi}_0^{\otimes 2}}_{\mathcal{C}}P_{i_1}^{\otimes 2}C_{1}^{\otimes 2}}_{C_1}&\equiv& \aver{\psi_{1}^{\otimes 2}}_{\mathcal{C}}\\\hspace{-1cm}&=&\frac{\tr{\Pi_2^{(\overline{i_1})}}}{\tr\Pi_2\tr\Pi_2}\Pi_2\hspace{0.2cm}\nonumber
\ea
Note that the above result is independent from the qubit $i_1$ on which the projector $P_{i_1}$ acts and from the basis state on which $P_{i_1}$ projects onto. Indeed $\tr \Pi_{2}^{(\overline{i_1})}=d(d+2)/4$ just depends on how many qubits we are not measuring. If we iterate the process we obtain:
\be 
\aver{\hat{\psi}_k^{\otimes 2}}_{\mathcal{C}}=\parent{\frac{\tr{\Pi_2^{(\overline{i_1})}}}{\tr\Pi_2}}^{k}\frac{\Pi_2}{\tr\Pi_2}
\ee 
Defining $N_k=\tr\parent{\hat{\psi}_k^{\otimes 2}}$, the average of $N_k$, can be easily derived from the previous results and we obtain 
\be
\aver{N_k}_{\mathcal{C}}=\parent{\frac{\tr{\Pi_2^{(\overline{i_1})}}}{\tr\Pi_2}}^{k}\label{nk}
\ee 
It is possible to calculate the subsystem purity for the non-normalized output state, that reads
\ba
\aver{\pur\hat{\psi}_{k,A}}_{\mathcal{C}}&=&\tr\left(S_A\aver{\psi_{k}^{\otimes 2}}_{\mathcal{C}}\right)\\&=&\parent{\frac{\tr{\Pi_2^{(\overline{i_1})}}}{\tr\Pi_2}}^{k}\frac{\tr\parent{S_A \Pi_2}}{\tr \Pi_2}\nonumber\\&=&\parent{\frac{\tr{\Pi_2^{(\overline{i_1})}}}{\tr\Pi_2}}^{k}\frac{d_A+d_B}{d_A d_B + 1} \nonumber
\label{purpsia}
\ea
which can be easily derived noting that $S=S_AS_B$, recalling that $S_A\equiv \tilde{S}_A\otimes \bbbone_{B}^{\otimes 2}$ and similarly for $S_B$.
\subsubsection{Calculations for $\aver{\hat{\psi}_k^{\otimes 4}}_{\mathcal{C}}$}
Let $B_{\theta}$ be the single qubit basis defined in Definition \ref{thetabasis}, then let $\tilde{P}^{(i)}$ be one of the projector on $B_\theta$ and $P_i\equiv\tilde{P}^{(i)}\otimes \bbbone^{(\bar{i})}$. We study first the average of $\hat{\psi}_k^{\otimes 4}$, for the first step we have \cite{leone2021quantum}:
\be
\aver{C_{0}^{\otimes 4}\hat{\psi}^{\otimes 4}C_{0}^{\dag\otimes 4}}_{C_0}\equiv\aver{\hat{\psi}_0^{\otimes 4}}_{\mathcal{C}}=c_0\Pi_4 Q + b_0 \Pi_4
\ee 
where $Q$ is defined in  \ref{Cliffhaar} and $\Pi_4\equiv(4!)^{-1}\sum_{\sigma\in S_4}T_{\sigma}$. Given $\psi=\ket{0}\bra{0}^{\otimes n}$ then 
\ba
c_0&=&\frac{1}{d}\parent{\frac{1}{D_{+}}+\frac{1}{D_{-}}}-\frac{1}{D_{-}}\\
b_0&=&\parent{1-\frac{1}{d}}\frac{1}{D_{-}}
\ea 
where $D_{+}\equiv\tr(Q\Pi_{4})$ and $D_{-}\equiv\tr(\Pi_4)-\tr(Q\Pi_4)$. Applying the first projector $P_{i_1}$ on the $i_1$-th qubit to $\aver{\hat{\psi}_0^{\otimes 4}}_{\mathcal{C}}$ we obtain 
\ba
\hspace{-1cm}P_{i_1}^{\otimes 4}\!\aver{\hat{\psi}_0^{\otimes 4}}_{\mathcal{C}}\!P_{i_1}^{\otimes 4}\hspace{-0.3cm}&=&\hspace{-0.2cm}P_{i_1}^{\otimes 4}\parent{c_0\Pi_4 Q + b_0 \Pi_4}P_{i_1}^{\otimes 4}\nonumber\\\hspace{-0.3cm}&=&\hspace{-0.2cm}c_0 \tr\parent{\tilde{P}^{(i_1)\otimes 4}Q^{(i_1)}}\tilde{P}^{(i_1)\otimes 4}\otimes Q^{\parent{\overline{i_1}}}\Pi_4^{\parent{\overline{i_1}}}\nonumber\\\hspace{-0.3cm}&+&\hspace{-0.2cm}b_0 \tilde{P}^{(i_1)\otimes 4}\otimes \Pi_4^{\parent{\overline{i_1}}}
\ea
where we used that $Q=Q^{(i_1)}\otimes Q^{(\overline{i_1})}$ as proved in \cite{leone2021quantum} and $T_{\rho}=T_{\rho}^{(i_1)}\otimes T_{\rho}^{(\overline{i_1})}$ and $T_{\rho}^{(i_1)}\tilde{P}^{(i_1)}=\tilde{P}^{(i_1)}$ and we defined $\Pi_4^{\overline{i_1}}=(4!)^{-1}\sum_{\sigma\in S_4}T_{\sigma}^{\parent{\overline{i_1}}}$. Applying another Clifford operator $C_1$ we obtain: 
\ba
\hspace{-1cm}\aver{\hat{\psi}^{\otimes 4}_1}_{\mathcal{C}}\hspace{-0.3cm}&=&\hspace{-0.3cm}\sum_{\rho,\sigma\in S_4}W_{\rho\sigma}^{+}c_{0}\tr\parent{\tilde{P}^{(i_1)\otimes 4}Q^{(i_1)}}\label{cliff1}\\&\times&\tr\parent{\tilde{P}^{(i_1)\otimes 4}\otimes Q^{\parent{\overline{i_1}}}\Pi_4^{\parent{\overline{i_1}}}QT_{\sigma}}QT_{\rho}\nonumber\\
\hspace{-1cm}\hspace{-0.3cm}&+&\hspace{-0.3cm}\sum_{\rho,\sigma\in S_4}W_{\rho\sigma}^{-}c_0\tr\parent{\tilde{P}^{(i_1)\otimes 4}Q^{(i_1)}}\nonumber\\&\times&\tr\parent{\tilde{P}^{(i_1)\otimes 4}\otimes Q^{\parent{\overline{i_1}}}\Pi_4^{\parent{\overline{i_1}}}Q^{\perp}T_{\sigma}}Q^{\perp}T_{\rho}\nonumber\\
\hspace{-1cm}\hspace{-0.3cm}&+&\hspace{-0.3cm}\sum_{\rho,\sigma\in S_4}W_{\rho\sigma}^{+}b_0\tr\parent{\tilde{P}^{(i_1)\otimes 4}\otimes \Pi_4^{\parent{\overline{i_1}}}QT_{\sigma}}QT_{\rho}\nonumber\\\hspace{-1.3cm}&+&\hspace{-0.3cm}\sum_{\rho,\sigma\in S_4}W_{\rho\sigma}^{-}b_0\tr\parent{\tilde{P}^{(i_1)\otimes 4}\otimes \Pi_4^{\parent{\overline{i_1}}}Q^{\perp}T_{\sigma}}Q^{\perp}T_{\rho}\nonumber
\ea 
defining $D_{+}^{(\overline{i_1})}=\tr\parent{\Pi_4^{(\overline{i_1})}Q^{(\overline{i_1})}}$ and $D_{4}^{(\overline{i_1})}=\tr\parent{\Pi_4^{\parent{\overline{i_1}}}}$, after some algebra we obtain:
\ba
\hspace{-1cm}\aver{\hat{\psi}^{\otimes 4}_1}_{\mathcal{C}}\hspace{-0.3cm}&=&\hspace{-0.3cm}\frac{1}{D_{+}}c_{0}\tr\parent{\tilde{P}^{(i_1)\otimes 4}Q^{(i_1)}}^{2}D_{+}^{\parent{\overline{i_1}}}Q\Pi_4\\\hspace{-1cm}\hspace{-0.3cm}&+&\hspace{-0.3cm}\frac{1}{D_{-}}c_0\tr\parent{\tilde{P}^{(i_1)\otimes 4}Q^{(i_1)}}D_{+}^{\parent{\overline{i_1}}} Q^{\perp}\Pi_4\nonumber\\\hspace{-1cm}\hspace{-0.3cm}&-&\hspace{-0.3cm}\frac{1}{D_{-}}c_0\tr\parent{\tilde{P}^{(i_1)\otimes 4}Q^{(i_1)}}D_{+}^{\parent{\overline{i_1}}}\tr\parent{\tilde{P}^{(i_1)\otimes 4}Q^{(i_1)}}\nonumber\\
\hspace{-1cm}\hspace{-0.3cm}&+&\hspace{-0.3cm}\frac{1}{D_+}b_0\tr\parent{\tilde{P}^{(i_1)\otimes 4}Q^{(i_1)}}D_{+}^{\parent{\overline{i_1}}}Q\Pi_4\nonumber\\\hspace{-1cm}\hspace{-0.3cm}&+&\hspace{-0.3cm}\frac{1}{D_{-}}b_0\parent{D_{4}^{\parent{\overline{i_1}}}-\tr\parent{\tilde{P}^{(i_1)\otimes 4}Q^{(i_1)}}D_{+}^{\parent{\overline{i_1}}}}Q^{\perp}\Pi_4\nonumber
\ea 
The above equation can be rearranged as 
\be 
\aver{\hat{\psi}^{\otimes 4}_1}_{\mathcal{C}}=\begin{pmatrix}Q \Pi_4 & \Pi_4\end{pmatrix}
\begin{pmatrix} I & J \\ K & L\end{pmatrix}\binom{c_0}{b_0}
\ee 
where the coefficients $I,J,K$ and $L$ are defined as
\ba
\hspace{-0.7cm}I\hspace{-0.3cm}&\equiv&\hspace{-0.3cm}D_{+}^{\parent{\overline{i_1}}}\tr\parent{\tilde{P}^{(i_1)\otimes 4}Q^{(i_1)}}\left(\frac{1}{D_+}\tr\parent{\tilde{P}^{(i_1)\otimes 4}Q^{(i_1)}}\right.\nonumber\\
\hspace{-0.7cm}\hspace{-0.3cm}&-&\hspace{-0.3cm}\left.\frac{1}{D_{-}}\parent{1-\tr\parent{\tilde{P}^{(i_1)\otimes 4}Q^{(i_1)}}}\right)\nonumber\\
\hspace{-0.7cm}J\hspace{-0.3cm}&\equiv&\hspace{-0.3cm}\frac{1}{D_+}\tr\parent{\tilde{P}^{(i_1)\otimes 4}Q^{(i_1)}}D_{+}^{\parent{\overline{i_1}}}\nonumber\\
\hspace{-0.7cm}\hspace{-0.3cm}&-&\hspace{-0.3cm}\frac{1}{D_{-}}\parent{D_{4}^{\parent{\overline{i_1}}}-\tr\parent{\tilde{P}^{(i_1)\otimes 4}Q^{(i_1)}}D_{+}^{\parent{\overline{i_1}}}}\label{coeffijkl}\\
\hspace{-0.7cm}K\hspace{-0.3cm}&\equiv&\hspace{-0.3cm}\frac{D_{+}^{\parent{\overline{i_1}}}}{D_{-}}\tr\parent{\tilde{P}^{(i_1)\otimes 4}Q^{(i_1)}}\parent{1-\tr\parent{\tilde{P}^{(i_1)\otimes 4}Q^{(i_1)}}}\hspace{-0.2cm}\nonumber\\
\hspace{-0.7cm}L\hspace{-0.3cm}&\equiv&\hspace{-0.3cm}\frac{1}{D_{-}}\parent{D_{4}^{\parent{\overline{i_1}}}-\tr\parent{\tilde{P}^{(i_1)\otimes 4}Q^{(i_1)}}D_{+}^{\parent{\overline{i_1}}}}\nonumber
\ea 
Reiterating this procedure one obtain at the end
\be 
\aver{\hat{\psi}^{\otimes 4}_k}_{\mathcal{C}}=\begin{pmatrix}Q \Pi_4 & \Pi_4\end{pmatrix}
\begin{pmatrix} I & J \\ K & L\end{pmatrix}^{k}\binom{c_0}{b_0}
\label{psioutputk}\ee 
Obtained the average of the fourth moment of the non-normalized output state we can calculate the fluctuations of $N_{k}$, the fluctuations of the subsystem purity, and the covariance between the purity and $N_k$. The fluctuations of $N_{k}$ are 
\ba
\Delta_{\mathcal{C}}N_{k}&\equiv&\aver{\tr\parent{\hat{\psi}^{\otimes 4}_k}}_{\mathcal{C}}-\aver{N_k}_{\mathcal{C}}^2\label{deltan} \\
&=&\begin{pmatrix}D_{+} & D_4\end{pmatrix}
\begin{pmatrix} I & J \\ K & L\end{pmatrix}^{k}\binom{c_0}{b_0}\nonumber\\&-&\parent{\frac{\tr{\Pi_2^{(\overline{i_1})}}}{\tr\Pi_2}}^{2k} \nonumber
\ea 
while the fluctuations of the subsystem purity are defined as
\ba
\hspace{-1cm}\Delta_{\mathcal{C}}\pur\hat{\psi}_{k,A}&\equiv&\aver{\pur^{2}\hat{\psi}_{k,A}}_{\mathcal{C}}-\aver{\pur\hat{\psi}_{k,A}}_{\mathcal{C}}^{2}\nonumber\\
&=&\tr\parent{T_{(12)(34)}^{A} \aver{\hat{\psi}^{\otimes 4}_k}_{\mathcal{C}}}\\&-&\parent{\frac{d_A+d_B}{d_A d_B +1}}^2\parent{\frac{\tr{\Pi_2^{(\overline{i_1})}}}{\tr\Pi_2}}^{2k}\label{deltapur}\nonumber\\
&=&\begin{pmatrix}D_{+} & D_{\pur}\end{pmatrix}
\begin{pmatrix} I & J \\ K & L\end{pmatrix}^{k}\binom{c_0}{b_0}\nonumber\\&-&\parent{\frac{d_A+d_B}{d_A d_B +1}}^2\parent{\frac{\tr{\Pi_2^{(\overline{i_1})}}}{\tr\Pi_2}}^{2k}\nonumber
\ea
where $D_{\pur}$ reads
\ba
\hspace{-1cm}D_{\pur}&\equiv& \tr\left(T_{(12)(34)}^{(A)}\Pi_4\right)\\
&=&(24)^{-1}\left(d_A^2 d_B^4 2 d_A^3 d_B^3 + 4 d_B^3 d_A+ 4 d_A^3 d_B \right.\nonumber\\&+&\left. 10 d_A^2 d_B^2 +  2 d_A d_B + d_A^4 d_B^2\right)\nonumber
\ea
The covariance between the purity and $N_k$ is given by
\ba
\hspace{-0.7cm}\operatorname{Cov}(\pur\hat{\psi}_{k,A},N_k)\hspace{-0.3cm}&=&\hspace{-0.3cm}\aver{\pur \hat{\psi}_{k,A}N_k}_{\mathcal{C}}-\aver{\pur\hat{\psi}_{k,A}}_{\mathcal{C}}\aver{N_k}_{\mathcal{C}}\nonumber\hspace{-0.3cm}\\\hspace{-0.3cm}&=&\hspace{-0.3cm}\aver{\tr(S_{A}\hat{\psi}_{k}^{\otimes 2})\tr(\hat{\psi}_{k}^{\otimes 2})}_{\mathcal{C}}\\\hspace{-0.3cm}&-&\hspace{-0.3cm}\aver{\tr\parent{S_{A}}\hat{\psi}_k^{\otimes 2}}_{\mathcal{C}}\aver{\tr\parent{\hat{\psi}^{\otimes 2}_k}}_{\mathcal{C}}\nonumber\\\hspace{-0.3cm}&=&\hspace{-0.3cm}
\begin{pmatrix}D_{+(12)}& D_{4(12)}\end{pmatrix}\begin{pmatrix}I & J\nonumber\\
K& L
\end{pmatrix}^{k}\binom{c_0}{b_0}\label{covnkpur}\\\hspace{-0.3cm}&-&\hspace{-0.3cm}\frac{d_A +d_B}{d_A d_B + 1}\parent{\frac{\tr{\Pi_2^{(\overline{i_1})}}}{\tr\Pi_2}}^{2k}  \nonumber
\ea 
where:
\ba
\hspace{-0.7cm}D_{+(12)}&\equiv&
\frac{1}{4!}\sum_{\rho\in  S_4}\tr(T_{(12)}^{(A)}Q^{(A)}T_{\rho}^{(A)})\tr(Q^{(B)}T_{\rho}^{(B)})\nonumber\\\hspace{-0.7cm}
D_{4(12)}&\equiv&\frac{1}{4!}\sum_{\rho\in  S_4}\tr(T_{(12)}^{(A)}T_{\rho}^{(A)})\tr(T_{\rho}^{(B)})
\ea
Thus we proved that to compute the average output state to the fourth tensor power $\aver{\psi_{k}^{\otimes 4}}_{\mathcal{C}}$, we just need to evaluate four coefficients $I,J,K,L$ in Eq. \eqref{coeffijkl} and take the $k$-th matrix power, cfr. \eqref{psioutputk}. Similarly to the case of $\aver{\psi_{k}^{\otimes 2}}$, the result is independent of the positioning of the measured qubit $i_\alpha$ at each iteration. We thus compute these coefficients for $\tilde{P}_{i}$ being one of the projectors of the basis $B_{\theta}$, introduced in Definition \ref{thetabasis}:
\be
\tr\parent{\tilde{P}^{(i_1)\otimes 4}Q^{(i_1)}}=\frac{7+\cos(4\theta)}{16}
\ee
note that this result is independent of the projector of the basis $B_\theta$ we choose; the reason behind it is explained in Remark \ref{remark1}. The four coefficients read:
\ba
\hspace{-0.7cm}I&=&\frac{(7+\cos(4\theta))(7d^2+21d-64+d(d+3)\cos(4\theta))}{1024(d^2-1)}\hspace{-0.3cm}\nonumber\\
\hspace{-0.7cm}J&=&\frac{3d(d-1)+d(d+3)\cos(4\theta)}{64(d^2-1)}\label{coeffbasistheta}\\
\hspace{-0.7cm}K&=&\frac{125+4\cos(4\theta)-\cos(8\theta)}{512(d^2-1)}\nonumber\\
\hspace{-0.7cm}L&=&\frac{(d+7)(d-1)-\cos(4\theta)}{16(d^2-1)}\nonumber
\ea
From here, one can easily compute $\aver{\psi_{k}^{\otimes 4}}_{\mathcal{C}}$ and consequently evaluate Eqs. \eqref{deltan}, \eqref{deltapur} and \eqref{covnkpur}.\qed

\subsection{The case $d_{A}=O(1)$}\label{app:bla}
In this section, we compute the average purity and purity fluctuations when $d_A=O(1)$ and $d_{B}=O(d)$. For sake of completeness, we insert the values of the average purity and purity fluctuations for the Clifford group and the full unitary group:
\ba
\hspace{-0.2cm}\aver{\pur(\psi_U)_A}_{U\in\haar}\hspace{-0.3cm}&=&\hspace{-0.2cm}\aver{\pur(\psi_C)_A}_{C\in \mathcal C(d)}=\frac{d_A+d_B}{d_A d_B +1}\nonumber\\
\hspace{-0.2cm}\Delta_{\haar}\pur(\psi_U)_{A}\hspace{-0.3cm}&=&\hspace{-0.3cm} \frac{2(d^2-d_A^2)(d_A^2-1)}{d_{A}^2(d+1)^2(d+2)(d+3)}=\Theta(d^{-2})\hspace{0.2cm}\nonumber\\
\hspace{-0.2cm}\Delta_{\mathcal{C}(d)}\pur(\psi_C)_{A}\hspace{-0.3cm}&=&\hspace{-0.3cm}\frac{\left(d_A^2-1\right) \left(d^2-d_A^2\right)}{(d+1)^2 (d+2) d_A^2}=\Theta(d^{-1}) 
\label{cliffordfluct}
\ea
The purity for the measurement doped circuits is equal to:
\be 
\aver{\pur \psi_{k,A}}_{\mathcal{C}}=\frac{d_A+d_B}{d_A d_B+1}+\Theta(k (d_A d)^{-1})
\ee 
while its purity fluctuations scales like: 
\be 
\Delta_{\mathcal{C}}\pur\psi_{k,A}=\Omega\left(\frac{d_A^2-1}{d_A^2}\left(\frac{7+\cos(4\theta)}{8}\right)^{2k}d^{-1}+\frac{d_A^2-1}{d_A^2}pd^{-2}\right)
\ee 
where $p$ is a constant independent from $k$. Looking at the above scalings, it is clear that we do not have loss of generality considering just the case $d_{A}=d_{B}=\sqrt{d}$. 

\section{Other proofs}
\subsection{Proof of Theorem \ref{th4}}\label{th4proof}
\subsubsection{$(\mathcal{M},\mathcal{C}_k)$-fold channel of order $2$}
Let $\mathcal{O}_2\in\mathcal{B}(\mathcal{H}^{\otimes 2})$ be a linear operator on $\mathcal{H}^{\otimes 2}$. Act on it with a global Clifford circuit $C_{0}^{\otimes 2}$ and take the average over $C_{0}$
\be
\aver{\mathcal{O}_2}_{C_0}=\sum_{\rho\in  S_2}a_{\rho}(\mathcal{O}_2)T_{\rho}
\ee
where $a_{\rho}(\mathcal{O}_2)=\sum_{\sigma\in  S_2}W_{\rho\sigma}T_{\sigma}$. Apply the map $\mathcal{M}$ on $\aver{\mathcal{O}_2}_{C_0}$
\be
\mathcal{M}^{\otimes 2}(\aver{\mathcal{O}_2}_{C_0})=\sum_{\rho\in  S_2}a_{\rho}(\mathcal{O}_2)\mathcal{M}^{\otimes 2}(T_{\rho})
\ee
then apply another Clifford operator $C_1^{\otimes 2}$ and take the average
\be
\aver{\mathcal{M}^{\otimes 2}(\aver{\mathcal{O}_2}_{C_0})}_{C_1}=\sum_{\rho,\sigma\in  S_4}\Xi_{\sigma\rho}a_{\rho}(\mathcal{O}_2)T_{\sigma}
\ee
where $\Xi_{\rho\sigma}=\sum_{\kappa}W_{\sigma\kappa}\tr(\mathcal{M}^{\otimes 2}(T_{\rho})T_{\kappa})$. At the $k$-th iteration one gets the desired result
\be
\Phi^{(2)}_{\mathcal{C},k}(\mathcal{O}_2)=\sum_{\rho,\sigma\in  S_4}(\Xi^{k})_{\sigma\rho}a_{\rho}(\mathcal{O}_2)T_{\sigma}
\ee
\subsubsection{$(\mathcal{M},\mathcal{C}_k)$-fold channel of order $4$}
Let $\mathcal{O}_{4}\in\mathcal{B}(\mathcal{H}^{\otimes 4})$. Act with a Clifford circuit $C_{0}^{\otimes 4}$ and take the average on $C_0$:
\ba
\Phi^{(4)}_{(\mathcal{M},\mathcal{C}_0)}&\equiv& \aver{C_{0}^{\otimes 4}\mathcal{O}_4C_{0}^{\dag\otimes 4}}_{C_0}\\&=&\sum_{\rho,\sigma\in  S_4}(c_{\rho}^{(0)}(\mathcal{O}_4)Q+b_{\rho}^{(0)}(\mathcal{O}_4))T_{\rho}\nonumber
\label{initial}
\ea
the proof of Eq. \eqref{initial} can be found in \cite{leone2021quantum} and $c_{\rho}^{(0)}(\mathcal{O}_4)$ and $b_{\rho}^{(0)}(\mathcal{O}_4)$ are shown in Eq. \eqref{coefficientsinitialth}. Let $\mathcal{M}$ be a CPTP quantum map, the action of $\mathcal{M}^{\otimes 4}$ on $\aver{\mathcal{O}_4}_{C_0}$ reads
\ba
\mathcal{M}^{\otimes 4}\parent{\aver{\mathcal{O}_4}_{C_0}}\hspace{-0.3cm}&=&\hspace{-0.3cm}\sum_{\rho,\sigma\in  S_4}c_{\rho}^{(0)}(\mathcal{O}_4)\mathcal{M}^{\otimes 4}(QT_{\rho})\nonumber\\\hspace{-0.3cm}&+&\hspace{-0.3cm}\sum_{\rho,\sigma\in  S_4}b_{\rho}^{(0)}(\mathcal{O}_4)\mathcal{M}^{\otimes 4}(T_{\rho})
\ea
Act with another Clifford operator $C_{1}^{\otimes 4}$ and take the average:
\ba
\Phi_{(\mathcal{M},\mathcal{C}_1)}^{(4)}&\equiv&\aver{\mathcal{M}^{\otimes 4}\parent{\aver{\mathcal{O}_4}_{C_0}}}_{C_1}\\&=&\sum_{\rho,\sigma\in  S_4}(c_{\rho}^{(1)}(\mathcal{O}_4)Q+b_{\rho}^{(1)}(\mathcal{O}_4))T_{\rho}\nonumber
\ea
\ba
\hspace{-0.7cm}c_{\kappa}^{(1)}(\mathcal{O}_4)\hspace{-0.2cm}&=&\hspace{-0.3cm}\sum_{\rho,\sigma\in  S_4}W^{+}_{\sigma\kappa}c^{(0)}_{\rho}(\mathcal{O}_4)\tr(\mathcal{M}^{\otimes 4}(QT_{\rho})QT_{\sigma})\nonumber\\
\hspace{-0.2cm}&+&\hspace{-0.3cm}\sum_{\rho,\sigma\in  S_4}W^{+}_{\sigma\kappa}b^{(0)}_{\rho}(\mathcal{O}_4)\tr(\mathcal{M}^{\otimes 4}(T_{\rho})QT_{\sigma})\label{ckappa}\\
\hspace{-0.2cm}&-&\hspace{-0.3cm}\sum_{\rho,\sigma\in  S_4}W^{-}_{\sigma\kappa}c^{(0)}_{\rho}(\mathcal{O}_4)\tr(\mathcal{M}^{\otimes 4}(QT_{\rho})Q^{\perp}T_{\sigma})\hspace{0.6cm}\nonumber\\
\hspace{-0.2cm}&-&\hspace{-0.3cm}\sum_{\rho,\sigma\in  S_4}W^{-}_{\sigma\kappa}b^{(0)}_{\rho}(\mathcal{O}_4)\tr(\mathcal{M}^{\otimes 4}(T_{\rho})Q^{\perp}T_{\sigma})\nonumber\\
\hspace{-0.7cm}b_{\kappa}^{(1)}(\mathcal{O}_4)\hspace{-0.2cm}&=&\hspace{-0.3cm}\sum_{\rho,\sigma\in  S_4}W^{-}_{\sigma\kappa}c^{(0)}_{\rho}(\mathcal{O}_4)\tr(\mathcal{M}^{\otimes 4}(QT_{\rho})Q^{\perp}T_{\sigma})\nonumber\\\hspace{-0.2cm}&+&\hspace{-0.3cm}\sum_{\rho,\sigma\in  S_4}W^{-}_{\sigma\kappa}b^{(0)}_{\rho}(\mathcal{O}_4)\tr(\mathcal{M}^{\otimes 4}(T_{\rho})Q^{\perp}T_{\sigma})\label{bkappa}
\ea
Therefore, let us define the following $24\times 24$ matrices:
\ba
\hspace{-1.1cm}M_{\kappa\rho}\hspace{-0.3cm}&=&\hspace{-0.3cm}\sum_{\sigma\in S_4}(W^{+}_{\sigma\kappa}\tr(\mathcal{M}^{\otimes 4}(QT_{\rho})QT_{\sigma})-W^{-}_{\sigma\kappa}\tr(\mathcal{M}^{\otimes 4}(QT_{\rho})Q^{\perp}T_{\sigma}))\hspace{-1cm}\nonumber\\
\hspace{-1.1cm}N_{\kappa\rho}\hspace{-0.3cm}&=&\hspace{-0.3cm}\sum_{\sigma\in S_4}(W^{+}_{\sigma\kappa}\tr(\mathcal{M}^{\otimes 4}(T_{\rho})QT_{\sigma}))-W^{-}_{\sigma\kappa}\tr(\mathcal{M}(T_{\rho})Q^{\perp}T_{\sigma})\hspace{-1cm}\nonumber\\
\hspace{-1.1cm}O_{\kappa\rho}\hspace{-0.3cm}&=&\hspace{-0.3cm}\sum_{\sigma\in S_4}W^{-}_{\sigma\kappa}\tr(\mathcal{M}^{\otimes 4}(QT_{\rho})Q^{\perp}T_{\sigma})\\
\hspace{-1.1cm}P_{\kappa\rho}\hspace{-0.3cm}&=&\hspace{-0.3cm}\sum_{\sigma\in S_4}W^{-}_{\sigma\kappa}\tr(\mathcal{M}^{\otimes 4}(T_{\rho})Q^{\perp}T_{\sigma})\nonumber
\ea
Then Eq.\eqref{ckappa} and Eq.\eqref{bkappa} read:
\ba
\hspace{-0.7cm}c_{\kappa}^{(1)}(\mathcal{O}_4)\hspace{-0.3cm}&=&\hspace{-0.3cm}\sum_{\rho\in S_4}M_{\kappa\rho}c_{\rho}^{(0)}(\mathcal{O}_4)+N_{\kappa\rho}b_{\rho}^{(0)}(\mathcal{O}_4)\hspace{0.2cm}\\
\hspace{-0.7cm}b_{\kappa}^{(1)}(\mathcal{O}_4)\hspace{-0.3cm}&=&\hspace{-0.3cm}\sum_{\rho\in S_4}O_{\kappa\rho}c_{\rho}^{(0)}(\mathcal{O}_4)+P_{\kappa\rho}b_{\rho}^{(0)}(\mathcal{O}_4)\hspace{0.2cm}
\ea
The latter are recurrence relations which can be easily generalize to $k$ iterations:
\ba
\hspace{-0.7cm}c_{\kappa}^{(k)}(\mathcal{O}_4)\hspace{-0.3cm}&=&\hspace{-0.3cm}\sum_{\rho\in S_4}M_{\kappa\rho}c_{\rho}^{(k-1)}(\mathcal{O}_4)+N_{\kappa\rho}b_{\rho}^{(k-1)}(\mathcal{O}_4)\hspace{1cm}\\
\hspace{-0.7cm}b_{\kappa}^{(k)}(\mathcal{O}_4)\hspace{-0.3cm}&=&\hspace{-0.3cm}\sum_{\rho\in S_4}O_{\kappa\rho}c_{\rho}^{(k-1)}(\mathcal{O}_4)+P_{\kappa\rho}b_{\rho}^{(k-1)}(\mathcal{O}_4)\hspace{1cm}
\ea
with:
\be
\Phi_{(\mathcal{M},\mathcal{C}_k)}=\sum_{\rho,\sigma\in  S_4}c_{\rho}^{(k)}(\mathcal{O}_4)QT_{\rho}+b_{\rho}^{(k)}(\mathcal{O}_4)T_{\rho}
\ee
To solve the recurrence relation, define the column vectors $\mathbf{c}^{(1)}$ and $\mathbf{b}^{(1)}$ with components $c_{\rho}^{(k)}$ and $b_{\rho}^{(k)}$, $\rho=1,\dots, 24$ respectively, and write the above relations in matrix form:
\ba
\bold{c}^{(k)}&=&M\bold{c}^{(k-1)}+N\bold{b}^{(k-1)}\\
\bold{b}^{(k)}&=&O\bold{c}^{(k-1)}+P\bold{b}^{(k-1)}
\ea
Then define $\bold{C}^{(k)}=\binom{\bold{c}^{(k)}}{\bold{b}^{(k)}}$:
\be
\bold{C}^{(k)}=\binom{\bold{c}^{(k)}}{\bold{b}^{(k)}}=\begin{pmatrix}
M& N \\
O& P\\
\end{pmatrix}^{k}\binom{\bold{c}^{(0)}}{\bold{b}^{(0)}}
\label{powermatrixk}
\ee
This concludes the proof. \qed

\subsection{Proof of Proposition \ref{corollarydbstates2and4}\label{propinfinitestatesproof}}
\subsubsection{$(\mathcal{D}_{\theta},\mathcal{C}_k)$-fold channel of $\psi^{\otimes 2}$}
From Theorem \ref{th4} we need to compute the matrix $\Xi$, whose components read $\Xi_{\rho\sigma}=\sum_{\sigma\in  S_2}W_{\sigma\kappa}\tr(\mathcal{M}^{\otimes 2}(T_{\rho})T_{\kappa})$, for $\mathcal{M}=\mathcal{D}_\theta\equiv\tilde{\mathcal{D}}_{\theta}^{(i)}\otimes \bbbone^{(\bar{i})}$. First the following factorization holds:
\be
\tr(\mathcal{D}_{\theta}^{\otimes 2}(T_{\rho})T_{\kappa})=\tr(\tilde{\mathcal{D}}_{\theta}^{\otimes 2}(T_{\rho}^{(i)})T_{\kappa}^{(i)})\tr(T_{\rho}^{(\bar{i})}T_{\kappa}^{(\bar{i})})
\ee
then we find:
\be
\Xi=\begin{pmatrix}
1&\frac{d}{2(d^2-1)}\\
0 & \frac{d^2-2}{2(d^2-1)}
\end{pmatrix}
\ee
the coefficients $a^{(0)}_{\rho}(\psi^{\otimes 2})=\frac{1}{d(d+1)}$. Taking the $k$-th matrix power, the $k$-th coefficients read:
\ba
a_{e}^{(k)}(\psi^{\otimes 2})&=&\frac{1}{d^2}-\frac{h^k}{d^2(d+1)}\nonumber\\
a_{(12)}^{(k)}(\psi^{\otimes 2})&=&\frac{h^k}{d^2(d+1)}
\ea
the asymptotic value for $k\rightarrow\infty$
\ba
a_{e}^{(\infty)}(\psi^{\otimes 2})&=&\frac{1}{d^2}\nonumber\\
a_{(12)}^{(\infty)}(\psi^{\otimes 2})&=&0
\ea
Thus the asymptotic state reads:
\be
\lim_{k\rightarrow\infty}\Phi^{(2)}_{\mathcal{D}_{\theta},\mathcal{C}_k}(\psi^{\otimes 2})=\frac{\bbbone^{\otimes 2}}{d^2}
\ee
\subsubsection{$(\mathcal{D}_{\theta},\mathcal{C}_k)$-fold channel of $\psi^{\otimes 4}$}\label{sec:4foldchannel}
From Theorem \ref{th4}, we need to compute the four $24\times 24$ matrices in Eq. \eqref{24by24matrices} with $\mathcal{M}\equiv \mathcal{D}_{\theta}$. Looking at Eq. \eqref{24by24matrices} we need to evaluate $\tr(\mathcal{M}^{\otimes 4}(QT_{\rho})QT_{\sigma})$ and similar terms. As proved in \cite{leone2021quantum}, the following factorization holds:
\be
\tr(\mathcal{M}^{\otimes 4}(QT_{\rho})QT_{\sigma})=\tr(\tilde{\mathcal{D}}_{\theta}^{\otimes 4}(Q^{(i)}T_{\rho}^{(i)})Q^{(i)}T_{\sigma}^{(i)})\Omega^{+(\bar{i})}_{\rho\sigma}
\ee
where $Q^{(i)}=\frac{1}{4}(I^{\otimes 4}+X^{\otimes 4}+Y^{\otimes 4}+Z^{\otimes 4})$, $I,X,Y,Z$ are the single qubit Pauli matrices and $\Omega^{+(\bar{i})}_{\rho\sigma}=\tr(T_{\rho}^{(\bar{i})}Q^{(\bar{i})}T_{\sigma}^{(\bar{i})})$. After some straightforward, but long algebra one could compute the $4$ matrices $M,N,O,P$ for $\mathcal{M}\equiv \mathcal{D}_{\theta}$. Then, for $\psi=\ket{0}\bra{0}^{\otimes n}$ we have\cite{leone2021quantum}:
\ba
c_{\rho}^{(0)}(\psi^{\otimes4})&=&\frac{1}{d}\parent{\frac{1}{D_{+}}+\frac{1}{D_{-}}}-\frac{1}{D_{-}}\\
b_{\rho}^{(0)}(\psi^{\otimes4})&=&\parent{1-\frac{1}{d}}\frac{1}{D_{-}}
\label{initialcoeff}
\ea
From here we can formally calculate the coefficients $c_{\rho}^{(k)}(\psi^{\otimes 4})$ and $b_{\rho}^{(k)}(\psi^{\otimes 4})$ through Eq. \eqref{powermatrixk}. For the case $\theta=\pi/2$, we can explicitly calculate these coefficients: 
\ba
\hspace{-1cm}c_{e}^{(k)}\hspace{-0.3cm}&=&\hspace{-0.3cm}\frac{(d^4+4d^2+192)f^k+(6d^3+29d^2+126d+168)h^k}{(4d^4)(d+1)(d+2)(d+4)}\nonumber\\\hspace{-0.3cm}&-&\hspace{-0.3cm}\frac{3(d^2+28)g^k }{4d^4(d+1)(d+2)}-\frac{3}{4d^4}\\
\hspace{-1cm}c_{(ab)}^{(k)}\hspace{-0.3cm}&=&\hspace{-0.3cm}\frac{(d^2+12)g^k}{4d^3(d+1)(d+2)}-\frac{(d^2+16)f^k}{2d^3(d+1)(d+2)(d+4)}\nonumber\\
\hspace{-0.3cm}&-&\hspace{-0.3cm}\frac{h^k}{2d^3(d+1)}\\
\hspace{-1cm}c_{(abc)}^{(k)}\hspace{-0.3cm}&=&\hspace{-0.3cm}\frac{f^k}{4(d+1)(d+2)(d+4)}\\
\hspace{-1cm}c_{(abcd)}^{(k)}\hspace{-0.3cm}&=&\hspace{-0.3cm}\frac{(d^2-12)g^k}{4d^3(d+1)(d+2)}-\frac{(3d^2-16)f^k}{2d^3(d+1)(d+2)(d+4)}\nonumber\\
\hspace{-0.3cm}&+&\hspace{-0.3cm}\frac{h^k}{2d^3(d+1)}\\
\hspace{-1cm}c_{(ab)(cd)}^{(k)}\hspace{-0.3cm}&=&\hspace{-0.3cm}\frac{1}{d^4}+\frac{(5d^2-16)f^k}{d^4(d+1)(d+2)(d+4)}\nonumber
\\\hspace{-0.3cm}&+&\hspace{-0.3cm}\frac{(d^2-7)h^k-7(d-2)g^k}{4d^4(d+1)}
\ea 
\ba
\hspace{-1cm}b_{(e)}^{(k)}\hspace{-0.3cm}&=&\hspace{-0.3cm}\frac{1}{d^4}-\frac{7h^k}{d^4(d+1)}+\frac{28 g^k}{d^4(d+1)(d+2)}\nonumber
\\\hspace{-0.3cm}&-&\hspace{-0.3cm}\frac{64 f^k}{d^4(d+1)(d+2)(d+4)}\\
\hspace{-1cm}b_{(ab)}^{(k)}\hspace{-0.3cm}&=&\hspace{-0.3cm}\frac{16 f^k}{d^3(d+1)(d+2)(d+4)}-\frac{6g^k}{d^3(d+1)(d+2)}\nonumber\\
\hspace{-0.3cm}&+&\hspace{-0.3cm}\frac{h^k}{d^3(d+1)}\\
\hspace{-1cm}b_{(abc)}^{(k)}\hspace{-0.3cm}&=&\hspace{-0.2cm}b_{(ab)(cd)}^{(k)}=\frac{g^k}{d^2(d+1)(d+2)}\nonumber\\
\hspace{-0.3cm}&-&\hspace{-0.3cm}\frac{4 f^k}{d^2(d+1)(d+2)(d+4)}\\
\hspace{-1cm}b_{(abcd)}^{(k)}\hspace{-0.3cm}&=&\hspace{-0.3cm}\frac{f^k}{d(d+1)(d+2)(d+4)}
\ea
here $(ab),(abc),\dots$ label the conjucacy classes of $ S_4$ and we defined the followings:
\be
f\equiv\frac{d^2-8}{8(d^2-1)}, \quad g\equiv\frac{d^2-4}{4(d^2-1)}, \quad h\equiv\frac{d^2-2}{2(d^2-1)}
\label{functionsfgh}
\ee
Taking the limit for $k\rightarrow\infty$ of the coefficients, one finds that the only non-zero ones are:
\ba
\hspace{-0.7cm}-\frac{1}{3}c_{e}^{(\infty)}\hspace{-0.3cm}&=&\hspace{-0.3cm}c_{(12)(34)}^{(\infty)}=c_{(13)(24)}^{(\infty)}=c_{(14)(23)}^{(\infty)}=\frac{1}{4d^4}\hspace{0.9cm}\\
\hspace{-0.7cm}b_{e}^{(\infty)}\hspace{-0.3cm}&=&\hspace{-0.3cm}\frac{1}{d^4}
\ea
Thus the asymptotic state reads:
\ba
\hspace{-1cm}\lim_{k\rightarrow\infty}\Phi^{(4)}_{\mathcal{D}_{\frac{\pi}{2}},C_{k}}(\psi^{\otimes 4})\hspace{-0.3cm}&=&\hspace{-0.3cm}\frac{Q}{4d^4}(-3\bbbone^{\otimes 4}+T_{(12)(34)}\nonumber
\\\hspace{-0.3cm}&+&\hspace{-0.3cm}T_{(13)(24)}+T_{(14)(23)})+\frac{\bbbone^{\otimes 4}}{d^4}\hspace{-0.3cm}\nonumber
\\&=&\hspace{-0.3cm}\frac{\bbbone^{\otimes 4}}{d^4}
\ea
Where we have used the fact that $Q=T_{(12)(34)}Q=T_{(14)(23)}Q=T_{(13)(24)}Q$.
\subsection{Proof of Lemma \ref{lemmainfinitestates}}\label{lemmainfinitestatesproof}
For $\theta=\pi/4$ we cannot explicitly calculate the coefficients in \eqref{ckbkth3}, we rather can find their asymptotic value for $k\rightarrow\infty$. Define:
\be
\mathcal{S}_{\theta}:=\begin{pmatrix}
M& N \\
O& P\\
\end{pmatrix}
\ee
for $\theta=\pi/2,\pi/4$ we find, by inspection, that $\mathcal{S}_{\theta}$ is diagonalizable having just one eigenvalue equal to $1$, while the others less then 1. Then, let $\mathcal{P}_{\theta}$ be the $48\times 48$ projector onto the eigenspace with eigenvalue $1$ of $\mathcal{S}_{\theta}$, we find that $\mathcal{P}_{\pi/2}=\mathcal{P}_{\pi/4}\equiv\mathcal{P}$, that reads:
\be
\mathcal{P}=\frac{1}{d^4}\begin{pmatrix}
P^{(a)}& P^{(b)}\\P^{(c)}&P^{(d)}
\end{pmatrix}
\ee
where $P^{(a)},P^{(b)},P^{(c)},P^{(d)}$ are four $24\times 24$ blocks:
\ba
P^{(a)}_{\rho\sigma}&=&\frac{1}{4}\begin{cases}-3\tr(QT_{\sigma}), \quad \rho=e\\
\tr(QT_{\sigma}), \quad \rho=(ij)(kl)\\
0, \quad \forall\rho\neq e, (ij)(kl)
\end{cases}\\
P^{(b)}_{\rho\sigma}&=&\frac{1}{4}\begin{cases}-3\tr(T_{\sigma}), \quad \rho=e\\
\tr(T_{\sigma}), \quad \rho=(ij)(kl)\\
0, \quad \forall\rho\neq e, (ij)(kl)
\end{cases}\\
P^{(c)}_{\rho\sigma}&=&\begin{cases}\tr(QT_{\sigma}), \quad \rho=e\\
0, \quad \forall\rho\neq e
\end{cases}\\
P^{(d)}_{\rho\sigma}&=&\begin{cases}\tr(T_{\sigma}), \quad \rho=e\\
0, \quad \forall\rho\neq e
\end{cases}
\ea
Taking the limit for $k\rightarrow\infty$
\be
\lim_{k\rightarrow\infty}\mathcal{S}_{\pi/2}^{k}=\lim_{k\rightarrow\infty}\mathcal{S}_{\pi/4}^{k}=\mathcal{P}
\ee
In order to find the asymptotic value for the coefficient we just need to project the initial coefficients given in Eq. \eqref{initialcoeff}:
\be
\binom{c_{\rho}^{(\infty)}}{b_{\rho}^{(\infty)}}=\mathcal{P}\binom{c_{\rho}^{(0)}}{b_{\rho}^{(0)}}
\ee
From here we conclude that:
\be
\lim_{k\rightarrow\infty}\Phi^{(4)}_{\mathcal{D}_{\frac{\pi}{2}},C_{k}}(\psi^{\otimes 4})=\lim_{k\rightarrow\infty}\Phi^{(4)}_{\mathcal{D}_{\frac{\pi}{4}},C_{k}}(\psi^{\otimes 4})=\frac{\bbbone^{\otimes 4}}{d^4}
\ee
the convergence rate of $\mathcal{S}_{\pi/4}^{k}$ towards $\mathcal{P}$ is given by the subdominant eigenvalue $h_{\pi/4}$ of $\mathcal{S}_{\pi/4}$: by inspection we find that the subdominant eigenvalue $h_{\pi/4}=h_{\pi/2}\equiv\frac{(d^2-2)}{2(d^2-1)}$; this implies that the convergence rate is the same for $\theta=\pi/2,\pi/4$, which concludes the proof. \qed

\end{document}